\documentclass[11pt]{article}

\usepackage{amsthm,amsmath,amssymb,array,color,colortbl,graphicx,multirow}

\usepackage{microtype}
\usepackage{comment}
\usepackage{algorithm}
\usepackage{balance}
\usepackage{xspace}
\usepackage{esvect}
\usepackage{wasysym}

\usepackage{thmtools}
\usepackage{thm-restate}

\newif\ifdraft
\drafttrue

\usepackage{wrapfig}
\usepackage{algorithm}
\usepackage[noend]{algorithmic}
\usepackage{enumitem}

\renewcommand\vec[1]{\tilde{#1}}

\newcommand{\wrkset}{\mathcal{W}}

\newcommand{\STAT}{\textsc{Stat}}

\newcommand{\ON}{\textsc{On}}

\newcommand{\Cost}{\mathrm{Cost}}
\newcommand{\RecCost}{\mathrm{adj}}

\newcommand{\netw}{N}
\newcommand{\demg}{G}

\newcommand{\BST}{\emph{BST}}

\newcommand{\insrt}{\mathrm{insert}}
\newcommand{\forward}{\mathrm{forward}}
\newcommand{\adjust}{\mathrm{adjust}}
\newcommand{\reset}{\mathrm{reset}}

\def\abs#1{\lvert #1 \rvert}

\newcommand{\card}[1]{\lvert #1\rvert}

\def\d{\mathrm{d}}

\def\A{\mathcal{A}}

\newtheorem{theorem}{Theorem}
\newtheorem{claim}{Claim}
\newtheorem{definition}{Definition}
\newtheorem{property}{Property}
\newtheorem{lemma}{Lemma}

\def\system{\emph{ReNet\xspace}}
\def\systems{\emph{ReNet}s\xspace}

\newcommand\chen[1]{\color{red} \textbf{Chen: #1 }\color{black}}
\newcommand\stefan[1]{\color{blue}\textbf{Stefan: #1 }\color{black}}

\title{%{\Large 
ReNets: Toward Statically Optimal\\
Self-Adjusting Networks}

\author{Chen Avin$^1$ \quad Stefan Schmid$^2$\\
{\small $^1$ Ben Gurion University, Israel \quad
$^2$ University of Vienna, Austria} 
}

\date{}

\begin{document}

\maketitle

\begin{abstract} 
This paper studies the design of \emph{self-adjusting}
networks whose topology dynamically adapts to the workload,
in an \emph{online} and \emph{demand-aware} manner. 
This problem is motivated by emerging 
optical technologies which allow to reconfigure
the datacenter topology at runtime.
Our main contribution is \emph{ReNet},
a self-adjusting network 
which maintains a balance between
the benefits and costs of reconfigurations.
In particular, we show that \emph{ReNets} are
\emph{statically optimal} for arbitrary sparse communication demands,
i.e., perform
at least as good as any fixed demand-aware 
network 
designed with a perfect knowledge of the \emph{future} demand.
%of the demand \emph{a priori}. %and match existing entropy-based lower bounds.
Furthermore, \emph{ReNets} provide 
\emph{compact} and \emph{local} routing,
 by leveraging ideas from
self-adjusting datastructures.
\end{abstract}

\maketitle

\section{Introduction}\label{sec:intro}

Modern datacenter networks rely on efficient
network topologies (based on fat-trees~\cite{clos}, hypercubes~\cite{bcube,mdcube}, or expander~\cite{xpander}
graphs) to provide a high connectivity at low cost~\cite{survey2017datacenter}.
These datacenter networks have in common that 
their topology is \emph{fixed} and \emph{oblivious} 
to the actual demand (i.e., workload
or communication pattern)
they currently serve.
% (e.g., workloads from 
%batch processing and streaming applications, 
%scale-out databases, etc.). 
Rather, they are designed for 
all-to-all
communication patterns, by ensuring 
properties such as full bisection bandwidth
or $O(\log{n})$ route lengths between \emph{any}
node pair in a constant-degree~$n$-node network.
However, demand-oblivious networks
can be inefficient for more \emph{specific}
demand patterns, as they usually arise in practice:
Empirical studies show that
traffic patterns in datacenters are often \emph{sparse
and skewed}~\cite{projector}, featuring much (spatial and temporal) locality.

This paper investigates algorithms for \emph{demand-aware} networks (DANs):
networks which provide shorter average route lengths
by accounting for locality in the demand
and locating frequently
communicating node pairs topologically closer.
Shorter routes can improve network performance  (e.g., latency)
and reduce costs (e.g.,
load, energy consumption).

DANs come in two flavors:
\emph{fixed} and \emph{self-adjusting}.
Fixed DANs can exploit \emph{spatial} locality
in the demand. It has recently been shown that a 
fixed DAN can provide average route lengths in the order
of the (conditional) \emph{entropy} of the demand~\cite{disc17}, which can be
much lower than the $O(\log{n})$ route lengths
provided by demand-oblivious networks \emph{for specific demands}.
However, fixed DANs 
require \emph{a priori} knowledge of the demand.

\emph{Self-adjusting} DANs do not require such knowledge
and can additionally exploit 
\emph{temporal} locality, by adapting to 
the demand in an \emph{online} manner.
The vision of such self-adjusting networks is enabled 
by emerging optical technologies which allow us to
\emph{reconfigure} the topology over time~\cite{helios,rotornet,mordia,firefly,projector}.

However, the design of self-adjusting DANs
is challenging:
while more frequent reconfigurations allow 
to adapt the topology to the demand in a more
fine-grained manner, such reconfigurations 
also come at a cost. 
Hence, an optimal tradeoff between the benefits and the
costs of such reconfigurations has to be found.
Further challenges are introduced by the online
nature of the problem and the lack of a priori 
knowledge about the demand. 

%\begin{wrapfigure}{r}{0.55\columnwidth}
\begin{figure}[t]
\centering
\includegraphics[width=.82\columnwidth]{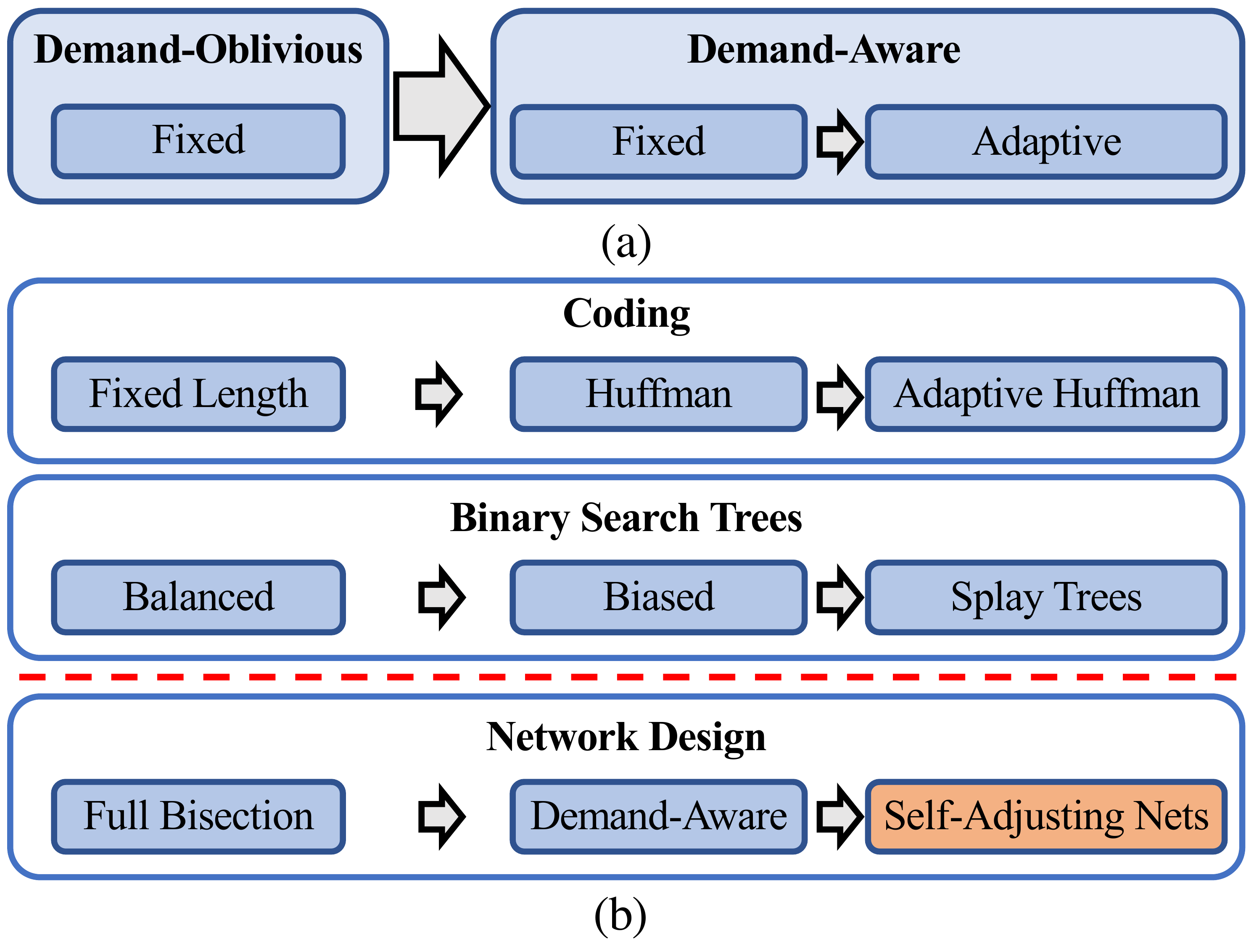}
\caption{Analogies of self-adjusting networks.}
%\vspace{-20pt}
\label{fig:schema}
%\end{wrapfigure}
\end{figure}

Ideally, a self-adjusting network provides
an optimal performance \emph{even in hindsight}:
despite the lack of information on the demand, 
the performance is at least as good as
the (fixed) demand-aware network, for
sufficiently long demands. 
This property is called \emph{static optimality}. 
Static optimality 
(and the related notion of \emph{regret
minimization})
is a strong notion of optimality and 
frequently used 
to evaluate algorithms based on limited
information, for example in the context of 
coding 
(e.g., dynamic Huffman codes \cite{Vitter1987}),
 self-adjusting datastructures
(e.g., splay trees ~\cite{splaytrees}),
or 
repeated games and
machine learning~\cite{auer2002nonstochastic,freund1999adaptive,littlestone1994weighted,bansal2003online}. 

\subsection{Analogy to Coding and Datastructures}

Our vision of self-adjusting networks, and in particular, 
the difference between demand-oblivious
and demand-aware networks (fixed and self-adjusting),
can be explained by an analogy.
This analogy will also provide an intuition why the entropy of the demand
is an important metric for self-adjusting networks.

In fact, we will provide two analogies: one to coding,
and one to datastructures, and in particular binary search trees (BSTs).
Figure~\ref{fig:schema} illustrates the analogy and evolution of
self-adjusting systems, from fixed demand-oblivious, over fixed
demand-aware, to self-adjusting demand-aware.

Simple, oblivious coding schemes which do not rely on
any knowledge on the input are based on
``fixed length'' coding: each symbol
is encoded using the same code length
(i.e., a logarithmic number of bits).
Without further information on the input, this is also optimal.
However,  
the performance of such an oblivious coding scheme 
may be far from optimal for a \emph{specific} input. 
In contrast to input-oblivious fixed-length coding, coding schemes
such as \emph{Huffman coding}~\cite{huffman1952method}
account for the frequency of the communicated symbols:
frequent symbols are assigned short codes,
infrequent symbols long codes.
Indeed, Huffman coding can result in much shorter
codes for specific inputs: code lengths 
are proportional to the \emph{entropy}
in the input.
A main drawback of such fixed approaches is that they
require a priori  knowledge about the (future) input, which
may not be available. 
Furthermore, they only allow to leverage \emph{spatial} locality, 
but not \emph{temporal} locality, i.e., change over time in the
input patterns/distribution. 
This can render such approaches inefficient under dynamic inputs,
or even impractical.
These limitations can be overcome by \emph{adaptive}
coding approaches, such as 
\emph{adaptive
Huffman codes} or arithmetic coding  
~\cite{knuth1985dynamic,Vitter1987}: these codes 
adjust to the input, in an \emph{online manner},
leveraging spatial \emph{and} temporal locality over time. 

Similar tradeoffs arise in 
the design of network topologies
(cf.~Figure~\ref{fig:schema}).
Traditional communication network
designs are \emph{demand-oblivious}
and \emph{fixed}, which can be suboptimal
under specific demand patterns.
In contrast, fixed DANs are optimized
toward a specific demand, which is assumed to 
be \emph{known}. Accordingly, the network can be designed
in such a way that traffic between \emph{frequently} 
communicating node pairs
is routed along \emph{shorter} paths, while traffic
between node pairs communicating \emph{infrequently}
is routed along \emph{longer} paths.
In other words, the network topology accounts for the \emph{entropy}
in the communication demand. %(see~\cite{disc17} for details).
However, the 
required knowledge about the (future) demand
may not be available, which makes 
adaptive (i.e., self-adjusting) DANs, like the one 
we present in this work, 
an attractive alternative:  
these networks can learn and adjust to the demand
\emph{in an online manner}.

Similar examples exist in the context of datastructures,
and in particular, Binary Search Trees (BSTs). As we will see,
BSTs play an important role in this paper in general.
Traditional BSTs are \emph{(demand-)oblivious}:
items are stored 
at distance $O(\log{n})$ from the root (on average), uniformly and independently
of their frequency. \emph{Demand-aware}, %~(statically-)optimized but
\emph{fixed} BSTs (a.k.a.~biased search trees)
such as~\cite{Mehlhorn75,knuth1971optimum,hu1971optimal,bent1985biased} 
account for the frequency 
of the accessed items: frequent items are stored close to the root, 
infrequent items are lower in the tree.
Finally, \emph{self-adjusting} BSTs, 
or self-adjusting demand-aware BSTs, 
such as \emph{splay trees}~\cite{splaytrees}
allow to adapt to the workload over time and account for 
spatial as well as temporal locality.

\subsection{Our Contributions}

The main contribution of this paper is
a self-adjusting demand-aware network called
$\system$ which:
(1) is provably statically optimal under sparse 
communication patterns, and therefore, 
as we will show, provides
entropy-proportional route lengths, 
without requiring
any knowledge of future demands; 
(2) is scalable in that it is of constant
degree and features \emph{compact 
routing} (i.e., constant-size forwarding tables); 
(3) supports \emph{local routing}, allowing us
to reconfigure 
networks seamlessly
while (4) relying on arbitrary addresses.

The $\system$ network relies on ideas
from self-adjusting datastructures. In particular
a $\system$ is based on a set of trees,
called \emph{ego-trees},
which are (dynamically) optimized for \emph{individual} 
nodes. A
$\system$ is then a union of all the ego-trees
of individual nodes, using algorithmic manoeuvres to
make sure that 
the degree (and routing tables) remain constant
at any time. 
The ego-tree of a given node 
stores the \emph{working set} of that node:
the set of its recent 
communication partners.

More specifically, the working set of 
each of these nodes is 
organized as a self-adjusting binary search tree
($\BST$).
While different types of such ego-trees can be used 
(e.g., Huffman trees, tango trees~\cite{tangotrees}, etc.), 
$\system$ uses \emph{splay trees}~\cite{splaytrees}. As we will see, this 
will result in desirable properties, such as compact and local routing. 

\subsection{Organization}

The remainder of this paper is organized
as follows.
We introduce our model and identify desirable properties
for self-adjusting demand-aware networks in 
Section~\ref{sec:model}.
Subsequently, we present
our algorithm (Section~\ref{sec:renet})
and its analysis (Section~\ref{sec:analysis}).
After reviewing related work in Section~\ref{sec:relwork},
we conclude and discuss future work
in Section~\ref{sec:conclusion}.
Some proofs and additional examples are deferred to the Appendix.

\section{Preliminaries and Model}
%~$\systems$}
\label{sec:model}

We consider a set~$V$ of~$n$ nodes 
$V=\{1,\ldots,n\}$ with unique  
but otherwise arbitrary addresses.
The communication \emph{demand} among these nodes
is described as a (finite or infinite) sequence 
$\sigma =
(\sigma_0, \sigma_1 \ldots)$ of \emph{communication
requests} where~$\sigma_t \subseteq V \times V$
is a set of source-destination communication pairs
$(u,v)  \in V \times V$ which communicate \emph{simultaneously}
at time $t$.
 The communication demand is
\emph{revealed in an online manner} and can be adversarial.

In order to serve this demand, the nodes~$V$ must
be inter-connected by a DAN~$\netw$, defined over the 
same set of nodes. 
In case of a \emph{self-adjusting} DAN,
 ~$\netw$ can also change over time, and we denote
  by~$\netw_t$ the network at time~$t$.
For scalability reasons and since
reconfigurable links may be costly and consume space,
the DAN must be chosen from 
the family of \emph{degree-bounded} topologies:
the networks $\netw$ considered in this paper
are required to be of constant degree at most $\Delta$. 

The route length to serve a request~$\sigma_t=(u,v)$ on the DAN,
is given by the hop distance~$\d_{\netw_t}(u,v)$ 
from~$u$ to~$v$, along the routing path chosen
by the algorithm over~$\netw_t$. If not specified otherwise, 
we assume shortest path routing.

We are interested in the fundamental tradeoff between the 
benefits of self-adjusting algorithms
(i.e., shorter routes) and their costs (namely reconfiguration costs). 
Let~$\A$ be an algorithm that given the request~$\sigma_t$ and the 
network~$\netw_t$ at time~$t$, transforms the current network 
to~$\netw_{t+1}$ at time~$t+1$. 
%We introduce a most simple,
%linear
%cost model that captures this tradeoff. Concretely, 
The cost of the network reconfiguration
at time~$t$ is given by the number of link changes performed to
change~$\netw_t$ to~$\netw_{t+1}$; when~$\A$ is clear from the context, we
will simply write this cost as~$\RecCost(\netw_{t},\netw_{t+1})$. 
Recall that an algorithm incurs a communication cost to serve 
request~$\sigma_t=(u,v)$, 
which depends on the hop distance~$\d_{\netw_t}(\sigma_t)=\d_{\netw_t}(u,v)$ 
from~$u$ to~$v$ in $\netw_t$.
\begin{definition}[\bf Average and Amortized Cost]
Given an algorithm~$\A$, an initial network~$\netw_0$, a 
distance function
$\d_N(\cdot)$, and a sequence~$\sigma=(\sigma_0,
\sigma_1 \ldots \sigma_{m-1})$ of communication requests over time,
we define the \emph{(average) cost} incurred by~$\A$ as:
%we study the following cost metric:
\begin{align}
\Cost( \A, \netw_0, \sigma) = \frac{1}{m}\sum_{t=0}^{m-1} (\d_{N_t}(\sigma_t) +
\RecCost(\netw_{i},\netw_{i+1}))
\end{align}
The \emph{amortized
cost} of~$\A$ is defined as the worst possible cost of~$\A$
 over all initial networks~$\netw_0$ and all sequences~$\sigma$, i.e.,
$\max_{\netw_0,\sigma}\Cost( \A, \netw_0, \sigma)$.
\end{definition}

In addition to the DAN,
nodes can also communicate over a demand-\emph{oblivious}
network: reconfigurable datacenter networks are usually \emph{hybrid}~\cite{projector},
connecting fixed (electric) switches with reconfigurable (optical) switches.
The demand-oblivious network plays a minor role in this paper, and is only used
to exchange control information (e.g., discover new neighbors).
Route lengths on the demand-oblivious networks cost
$D$ per request, where~$D$ is a parameter: e.g., $D$ is the 
diameter of the demand-oblivious network, $D=\Theta(\log{n})$ in constant-degree networks.

A main challenge faced by self-adjusting DANs 
is that information about (future) demand may not
be available. 
Our goal is to design algorithms for 
\emph{statically optimal} networks:
\begin{property}[\bf Static Optimality]
Let~$\STAT$ be an optimal static algorithm with perfect
knowledge of the demand~$\sigma$,
and let~$\ON$ be an online algorithm producing a sequence
of degree-bounded networks (i.e., the maximum degree
is at most $\Delta$). We say that~$\ON$ is \emph{statically
optimal}
if, for sufficiently long communication patterns~$\sigma$:
$$
\rho = \max_{\sigma} 
\frac{\Cost(\ON, \emptyset, \sigma)}{\Cost(\STAT,\netw^*, \sigma)} 
$$
\noindent is \emph{constant}.
Here,~$\netw_0 = \emptyset$ is the empty network from which ~$\ON$ starts,
and~$\netw^*$ is the statically optimal degree-bounded network for $\sigma$. 
In other words,~$\ON$'s cost is at most a 
constant factor higher than~$\STAT$'s in the worst case.
%\chen{formally, we either need to take~$\sigma$ to infinity or say it is large enough 
%and talk about a constant instead~$O(1)$.}
%\end{itemize}
\end{property}

We conclude this section with some definitions.
We first note that we can think of the entire sequence
$\sigma$, or a subsequence 
$\sigma'\subseteq \sigma$,
as a directed and weighted 
\emph{demand graph}~$\demg(\sigma')=(V(\sigma'),E(\sigma'))$. 
Here, the node set~$V$ of~$G$ is given by the set of nodes participating in 
$\sigma'$, i.e.,~$V(\sigma') = \{ v : v \in \sigma' \}$, and the set of directed edges~~$E$ is given by 
$E(\sigma') = \{\sigma'_t : \sigma'_t \in \sigma']\}$.  The weight~$w(e)$ of each directed edge~$e=(u,v)\in E$  is the frequency 
$f(u,v)$ of the request from~$u$ to~$v$ in~$\sigma'$,
where~$\sum_{u,v \in V(\sigma')} f(u,v)=1$.
We are interested in sparse communication patterns:
\begin{definition}[\bf~$(c,\delta)$-sparse Communication]\label{def:sparse}
We call a communication demand~$\sigma$ 
\emph{$(c,\delta)$-sparse} if and only if
any subsequence~$\sigma'$ of~$\sigma$ 
of length~$|\sigma'|\leq \delta$, 
involves no more than~$c\cdot n$ unique communication pairs where $c$ is a constant and $\delta$ is a function of $n$.
That is,~$\sigma'$ implies a 
  \emph{sparse demand graph} 
 ~$\demg(\sigma')=(V,E(\sigma'))$ of average
  degree~$2 \abs{E(\sigma')} / n \leq 2c$. 
\end{definition}
Note that for~$\delta=\infty$, the entire communication pattern
$\sigma$ needs to be sparse. For ~$\delta \le cn$, the constraint
is trivial.

We define the \textbf{\emph{entropy of the demand}}~$\sigma=(\sigma_1,\sigma_2,\ldots,
\sigma_m)$ to be served by a communication network. 
Recall that~$\sigma_t=(\hat{X},\hat{Y})$ describes a source-destination pair
(e.g., a client~$x$ communicating to a server~$y$). 
We will interpret~$\sigma$ as a \emph{joint empirical frequency 
distribution}~$(\hat{X},\hat{Y})$, where~$\hat{X}$ is the empirical frequency of the sources 
and~$\hat{Y}$ is the empirical frequency of the destinations.
In particular, in the following, the term \emph{entropy}
will refer to the \emph{empirical
entropy} of~$\sigma$, i.e., the entropy implied by
the communication \emph{frequencies}. 

More formally, let~$\hat{X}_\sigma =
\{f(x_1,\cdot), \dots , f(x_n,\cdot)\}$ be the empirical %entropy measure of the
frequency distribution of the \emph{communication sources} (origins)
occurring in the communication sequence~$\sigma$, 
i.e.,~$f(x_i,\cdot)$ is
the frequency with which a node~$x_i$ appears as 
a source in the sequence:
$f(x_i,\cdot) = (\# x_i \: \text{ is a source in }\: \sigma)/m$, when $m= \abs{\sigma}$ is the length of $\sigma$.
We omit~$\sigma$ in~$\hat{X}_\sigma$ when it is clear from the context.
The \emph{empirical entropy}~$H(\hat{X})$ is then defined as
$H(\hat{X}) = - \sum_{i=1}^n f(x_i)\log_2 f(x_i)$, where $f(x_i)$
is used as a shorthand for $f(x_i,\cdot)$. 
Similarly, we define
the empirical entropy of the \emph{communication destinations}
$H(\hat{Y})$: we
consider~$\hat{Y}_\sigma =
\{f(\cdot,y_1), \dots , f(\cdot,y_n)\}$ where
$f(\cdot,y_j)$ is
the frequency with which a node~$y_j$ appears as 
a destination in the sequence.
We use the normalization
$f(x \vert y) = f(x,y)/f( (\cdot, y)$.
The empirical \emph{joint} entropy~$H(\hat{X}, \hat{Y})$ 
 is defined as~$H(\hat{X}, \hat{Y}) = - \sum_{i,j} f(x_i, y_j)\log_2 f(x_i, y_j)$
 and  the empirical \emph{conditional} entropy 
~$H(\hat{X} \vert \hat{Y})$ 
which measures spatial locality  as 
~$H(\hat{X} \vert \hat{Y})= 
 - \sum_{j} f(y_j) \sum_{i} f(x_i \mid y_j )\log_2 f(x_i \mid y_j )$.
We may simply write~$H$ for the entropy if
the usage is given by the context.
By default, we will denote by~$H$ the entropy computed using the binary
logarithm; 
if a different logarithmic basis~$\Delta$ is used to compute the entropy,
we will explicitly write~$H_\Delta$.

It was recently shown that the 
\emph{conditional entropy} of the demand, and in particular 
$\max(H_\Delta(\hat{Y}|\hat{X}), H_\Delta(\hat{X}|\hat{Y})$ 
is a lower bound for the average route length in any (constant) 
degree-$\Delta$ bounded, \emph{fixed}
network~\cite{disc17}.\footnote{We note that the result 
in~\cite{disc17} is stated for 
the entropy and not the empirical entropy, however, 
the claim follows directly.}
This bound can be (asymptotically) matched if the demand is sparse and 
the demand~$\sigma$ is known \emph{a priory},
%~\cite{disc17} 
before designing the network. 
In contrast, in this work, we are interested in solutions
that match the conditional entropy lower bound, but for a demand $\sigma$ that is \emph{unknown} a priory. This makes 
the task much more challenging.

\section{Statically Optimal Self-Adjusting Networks}\label{sec:renet}
%\section{$\system$: Design Details}\label{sec:renet}

This section presents statically optimal algorithms for 
\systems,
self-adjusting networks of bounded degree 
which support the following additional desirable
properties. 
We first discuss desirable properties of such
self-adjusting networks, then present algorithmic
building blocks, and finally describe \systems' forwarding 
tables in details.

\subsection{Desirable Properties of Self-Adjusting Networks}

In order to ensure scalability, 
each node in a self-adjusting network should 
not only rely on at most a constant number
of reconfigurable links,
we would like to have an even stronger property: 
namely that the forwarding tables are of constant size, i.e., \emph{compact}~\cite{frederickson1988designing,thorup2001compact}. 
%\stefan{chen asks: should we give some examples like p2p?}
\begin{property}[\bf Compact Routing]\label{prop2}
A network supports
\emph{compact routing}
if the sizes of the nodes' forwarding tables are 
constant, i.e., independent of the network size.
\end{property}

A key challenge in the design of self-adjusting
networks is that topological reconfigurations
may negatively affect routing. 
A particularly attractive (but seemingly difficult to achieve) 
property for routing
in dynamic networks is \emph{local routing}:
\begin{property}[\bf Local Routing]\label{prop3}
A network provides \emph{local routing} if packets
can be forwarded based on local knowledge only. In particular,
topological changes should not entail the
global re-computation of routes. 
\end{property}

Furthermore, we do not want the compact and local routing 
properties to depend on any specific addressing scheme.
In particular, addresses can be flat and location-independent.
\begin{property}[\bf Arbitrary Addressing]\label{prop4}
Nodes can have arbitrary (but unique) addresses.
\end{property}

\subsection{Outline and Algorithmic Building Blocks}

We first describe the main ideas and building blocks of $\systems$.
Ideally, each node~$u\in V$ in~$\system$ connects 
\emph{directly}
to all its communication partners in~$\demg(\sigma)$,
 achieving an ideal
average route length of 1.
However, this is infeasible, as (1) the communication
partners are not known to~$u$ \emph{a priori} and
(2) a node~$u$ may have \emph{many} communication 
partners (even in an otherwise sparse demand graph), 
which would result in a high degree and 
large forwarding tables. 
To overcome this, ~$\system$ leverages several key concepts:

%\noindent \textbf{Concept 1: Working Set.} 
%\paragraph*{\bf Concept 1: Working Set.} 
\noindent \textbf{Concept 1 - Working Set.} 
Each node $u$ in~$\system$ 
keeps track of it recent active communication partners,
i.e., the so-called \emph{working set}~$\wrkset(u)$,
hoping to exploit temporal locality as they
are also likely to be relevant in the near future. 
The working set will be defined over the recent subinterval 
of~$\sigma$ that will be defined later. 

\noindent \textbf{Concept 2 - Small and Large Nodes.}
A node~$u$ in~$\system$ pursues one of
two different approaches to communicate with its recent communication
partners, depending on its working set size.
Towards this end, we define the \emph{size} of a node~$u$
to be the cardinality~$|\wrkset(u)|$ of~$u$'s working set.
We say that a node~$u$ is \emph{small} if the size of~$u$ 
is smaller than a parameter~$\theta$ and otherwise, a node is called \emph{large}.
For now assume $\theta$ is a constant which depends on the sparsity of the communication sequence (i.e., Definition \ref{def:sparse}), we will discuss the details later.
%or equal two times the possible average node degree in a subinterval of~$\sigma$,
%denoted by some constant parameter~$\theta$;
%otherwise, a node is called \emph{large}.
%
A small node will communicate to its communication partners
\emph{directly}; a large node \emph{indirectly},
by forwarding the traffic along its \emph{ego-tree}, which we explain next.

\noindent \textbf{Concept 3 - Ego-Tree/Ego-BST.} 
For large nodes~$u$, establishing links or
storing forwarding rules for each communication partner 
in~$\wrkset(u)$ is infeasible as it would result
in forwarding tables of non-constant size. 
Thus, in~$\system$, a large node~$u$ 
 organizes its communication partners~$\wrkset(u)$ 
in a self-adjusting \emph{ego-tree}: a \emph{(tree) network}
 optimized for just this
source. In particular, we propose to use a self-adjusting 
\emph{binary search tree} for the ego-tree of a large node $u$, short \emph{ego-}$\BST(u)$.
An important property of BSTs (fixed and self-adjusting) is that they naturally supports local and compact routing
for messages to or from the \emph(root) of the tree.
In particular, the (self-adjusting) \emph{ego-}$\BST(u)$, 
is used to efficiently \emph{store} and \emph{lookup} (i.e., forward to)
% (i.e., in our case: \emph{route to}) 
neighbors~$v\in \wrkset(u)$.
Each node that belongs to such a~\emph{ego-}$\BST(u)$ 
support the following interface:
\begin{itemize}
%\begin{description}
\item~\emph{ego-}$\BST(u).\insrt(v)$: insert~$v$ to~\emph{ego-}$\BST(u)$
\item~\emph{ego-}$\BST(u).\forward(v)$: forward packet toward~$v$ %on~$\BST(u)$
\item~\emph{ego-}$\BST(u).\adjust()$: local update of tree network datastructure
%\end{description}
\end{itemize}

%%%%%%%%%%%%%%%%%%%%%%%%%%%
\begin{figure*}[t]
\centering
\begin{tabular}{ccc}
\multirow{2}{*}[.3in]{\includegraphics[width=.25\textwidth]{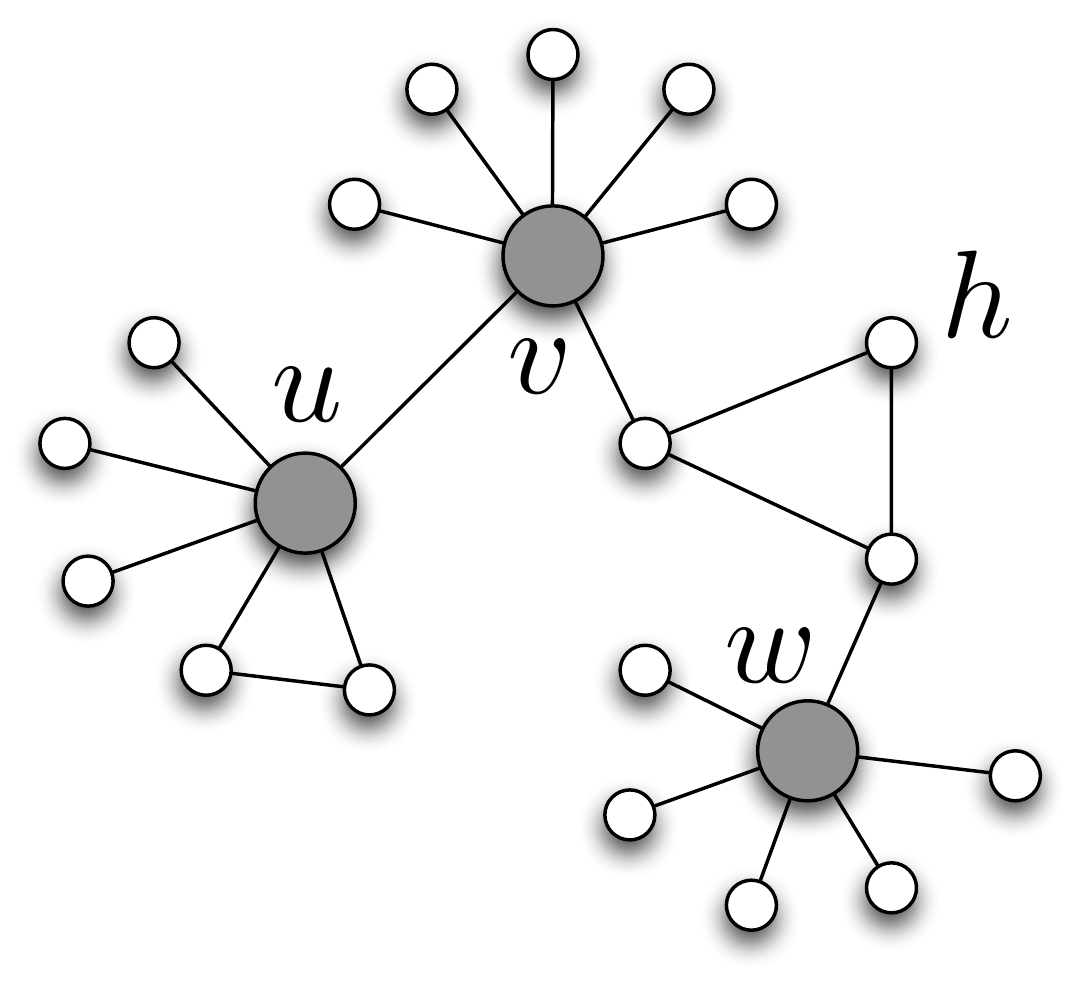}} &
\includegraphics[width=.4\textwidth]{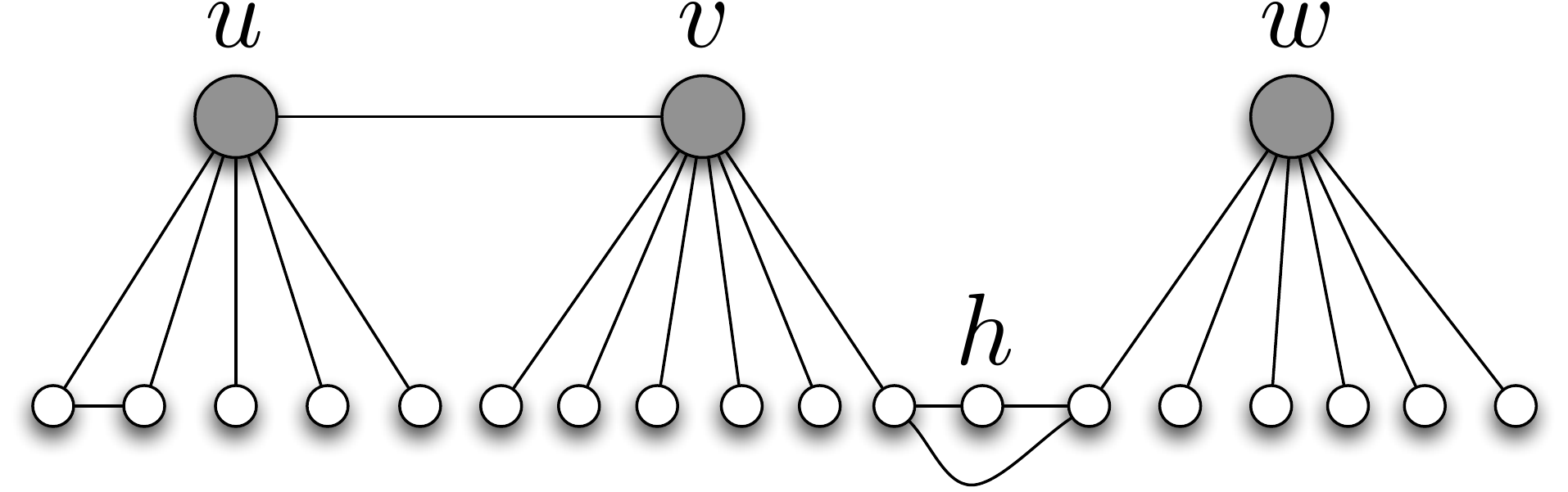} &
\multirow{2}{*}[.8in]{\includegraphics[width=.30\columnwidth]{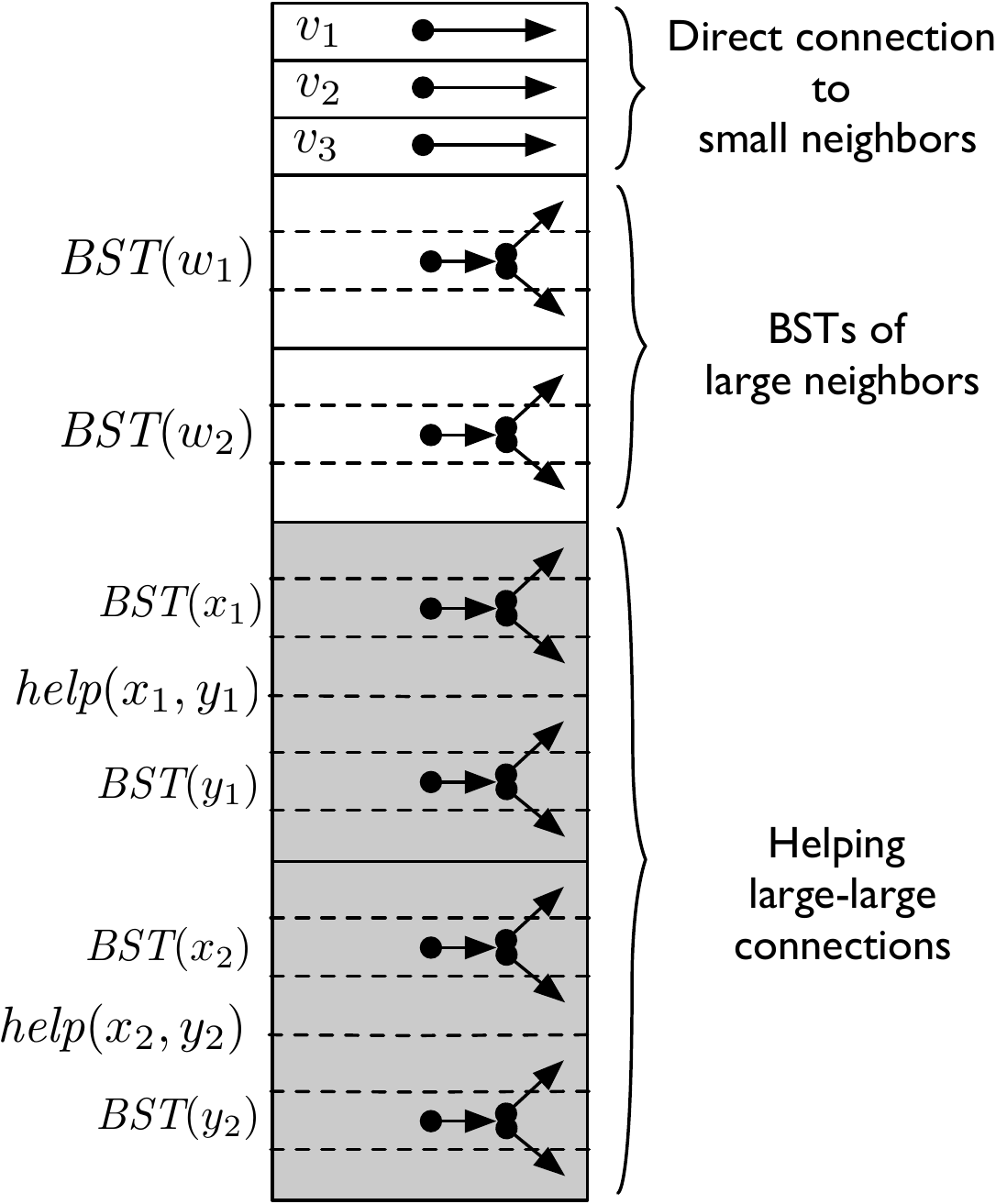}} \\
 & (b) & \\
&\includegraphics[width=.4\textwidth]{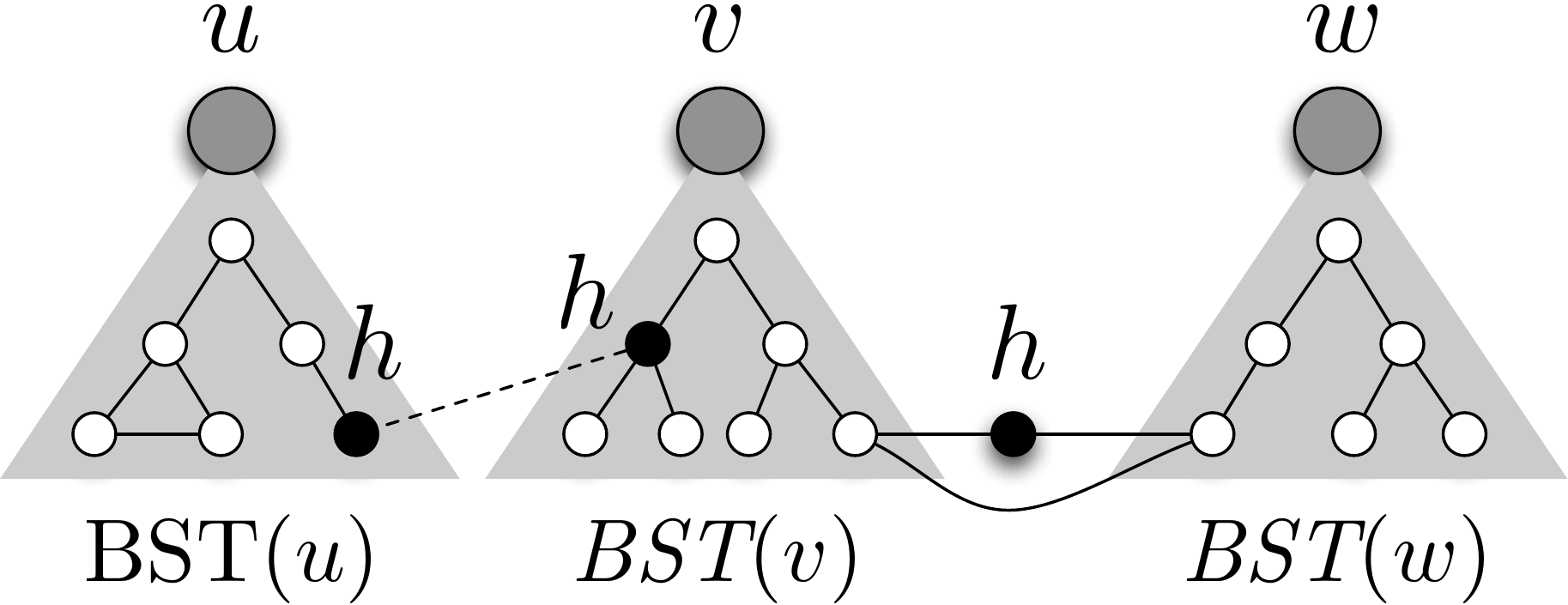} \\
(a) & (c) & (d)
\end{tabular}
\caption{Design principles of a (fixed and adaptive) \system: (a) The sparse demand graph
$\demg(\sigma)$. Nodes are divided into large (\emph{gray}, e.g.,~$u, v, w$) and small (\emph{white}, e.g.~$h$) nodes. (b) Hierarchical representation of the 
demand graph. Problematic edges are 
edges between two large nodes (e.g., ~$(u,v)$). (c) \system: every 
large node~$x$ has an~\emph{ego-}$\BST(x)$ connecting its \emph{working set}~$\wrkset(x)$. Every 
large-large edge, is routed with the help of a small node (acting as 
relay between~\emph{ego-}$\BST$s, \emph{black}). E.g.,~$h$ is the helping node for edge~$(u,v)$ and participates in~\emph{ego-}$\BST(u)$ as a relay 
node toward~$v$ and in~\emph{ego-}$\BST(v)$ as a relay toward~$u$. The resulting network has bounded degree. (d) Forwarding table for a \emph{small} node, see text.}
\label{fig:expamle-high}
\end{figure*}
%%%%%%%%%%%%%%%%%%%%%%%%%%%

\noindent \textbf{Concept 4 - Self-Adjustments.} 
$\systems$ perform two types of self-adjustments.
First,~$\systems$ are configured with \emph{self-adjusting}
versions of the underlying ego-BST, in particular, \emph{splay trees}.
In order to update the neighborhood structures and optimize
the network after a routing request, a node~$u$ makes use of
the~$\adjust()$
operation. For example in splay trees,
we issue a \emph{splay} operation on the tree network.
Second, ~$\systems$ keep track of the total size of the nodes' working sets. 
Once the total size, $\sum_v \wrkset(v)$, exceeds 
$n\cdot\theta/2 $, all working sets are \emph{cleared}
(in the spirit of \emph{flush-when-full} or 
\emph{marking} techniques known from
competitive paging~\cite{fiat1991competitive}). 
Such reset operations are necessary to follow temporal locality,
allowing the nodes to update the working sets and hence be 
able to adjust to changing demand patterns.
% that are sparse only in certain time-windows, 
%but \emph{dense} in the longer run.

%To this end, we use the~$\emph{reset()}$ \chen{flush()} operation
%offered by the datastructure.

\noindent \textbf{Concept 5 - Helper Nodes.} 
The problem with the approach described so far is that
while nodes in a single BST are of degree at most three
(\emph{parent}, \emph{left child}, 
\emph{right child}), a \emph{large} node~$v$ can
still appear in multiple trees beside it's own tree if he has \emph{large} nodes in its working set. 
%$\bigcup_{v\in \wrkset(w)}\emph{ego-}\BST(w)$.
Combined, these trees can induce a large forwarding table on~$v$,
and hence, an additional mechanism is needed to bound the degree.
To this end,~$\systems$ leverage small nodes to help
two communicating large nodes keep the forwarding table
small.
% similarly to techniques used for static network designs~\cite{disc17}. 
Concretely, as long as the average node size is smaller
than~$\theta/2$ %(so total number of unique edges is below $\theta \cdot n$)
, a $\system$ exploits small nodes (of size below $\theta$) 
as \emph{helper nodes}. For example a small mode $h$ may serve as \emph{relay} between two communicating
large nodes~$v$ and~$u$.  Node $h$ will appear in \emph{both} 
trees networks~\emph{ego-}$\BST(v)$ and~\emph{ego-}$\BST(u)$, in \emph{ego-}$\BST(v)$, $h$ will serve as a forwarder toward $u$ and in \emph{ego-}$\BST(u)$, as a forwarder toward $v$. See Figure \ref{fig:expamle-high} for example.
As we will discuss later, this not only allows us to bound the 
size of the forwarding table,
but also to preserve local routing.
%\stefan{chen: Concept 5, node names in definition of ego-BST are not correct} \chen{corrected. check}

\noindent \textbf{Concept 6 - Centralized Bookkeeping and Coordination.}
While reconfiguration is decentralized, 
bookkeeping and coordination is centralized in~$\system$.
This avoids complexities due to possible inconsistencies 
and is efficient: 
a network coordinator (e.g., an arbitrary node in the network) 
only needs to keep
track of which nodes are \emph{large} and which
nodes are \emph{small}. That is, nodes
inform the coordinator when they need to add a new partner to 
their working sets.
Given this information, the coordinator can assign
helper nodes upon request, in an event-driven manner. 
%\stefan{Helper nodes can be picked arbitrarily (e.g., randomly),
%and found leveraging the bookkeeping and routing.} \chen{let's discuss}
If no helper nodes are left, the coordinator
schedules a~$\reset()$ operation that clears all working sets for all nodes, and set their size to \emph{small}.
Such a reset can be done using a spanning tree, at linear cost.
%To improve efficiency, the nodes may implement the reset lazily, e.g., using %a marking algorithm~\cite{fiat1991competitive}. \chen{do we need this}.
%That said, communication to the controller entails the usual 
%costs and hence should be minimized. 

%\begin{figure}[t]
%\centering
% \begin{tabular}{c}
% \hspace{35pt} \includegraphics[width=.58\columnwidth]{tablelarge} \\
%(a) The forwarding table of a large node~$u$ \\ \\
%\includegraphics[width=.9\columnwidth]{tablesmall} \\
%(b)  The forwarding table of a small node~$u$
%\end{tabular}
%\caption{Forwarding tables of large and small nodes}
% \label{fig:rtable}
%\end{figure}

\begin{comment} 
 \begin{center}
  \begin{tabular}{| c | c | c | }
    \hline
     \emph{to} & \emph{direct neighbor?} & \emph{action} \\ 
    ~$w_1~$ & yes & \emph{forward directly} \\ 
    ~$w_2~$ & no & \emph{send to BST($w_2$)} \\ 
    ~$\ldots$ &~$\ldots$ &~$\ldots$ \\ 
         \hline
  \end{tabular}
\end{center}
\end{comment} 

%\begin{figure*}[t]
%\centering \label{fig:rtable}
%\includegraphics[width=1\textwidth]{rtab-1}\\
%\includegraphics[width=1\textwidth]{rtab-2}
%\end{figure*}

Figure~\ref{fig:expamle-high} illustrates some of the concepts introduced above, such as small and large nodes, working set,~$\BST$, and helping nodes.
% \stefan{describe more, make example bigger with more time steps}

\subsection{Details of Forwarding Table and Reconfiguration Algorithm}

With these intuitions in mind, we now present 
the network reconfiguration and forwarding algorithms
underlying~$\systems$ in details.

\subsubsection{Forwarding Table}

Each node~$u\in V$ maintains the following forwarding table
of given constant size~$6\theta$ (details later)
whose content may change over time.
If~$u$ is \emph{small},~$u$'s forwarding table contains (Figure~\ref{fig:expamle-high} (d)):
%see Figure~\ref{fig:rtable}:
%\begin{itemize}
%\setlength{\itemindent}{-15pt}
%\item If~$u$ is \emph{small},~$u$'s forwarding table contains (Figure~\ref{fig:expamle-high} (d)):
	\begin{itemize}
	\setlength{\itemindent}{0pt}
		\item A set~$S(u)=\{v_1,v_2,\ldots\}$ of \emph{small} neighbors of~$u$. For each $v_i$
		$u$ has a \emph{direct, physical} link (port) toward ~$v_i$.
		\item A set~$L(u)=\{\emph{ego-}\BST(w_1),\emph{ego-}\BST(w_2),\ldots\}$ of \emph{ego-BSTs} 
		of \emph{large} neighbors of~$u$.  In each of these trees, say~\emph{ego-}$\BST(w)$,~$u$ will participate and forward messages 
		toward/from the root of the tree,~$w$. Each such tree requires
		 three entries in the forwarding table, and three \emph{physical 
		 ports} (i.e., direct links): %For destination~$w$, 
		 (1) a forward entry to the parent of $u$ in~\emph{ego-}$\BST(w)$, 
		 (2) a forward entry %range of destination IDs for forwarding
		 to the left child of $u$ in~\emph{ego-}$\BST(w)$, and 
		 (3) similarly for the right child in~\emph{ego-}$\BST(w)$.
		\item A set~$H(u)=\{(x_1,y_1),(x_2,y_2),\ldots\}$ of pairs 
		of \emph{large} nodes~$x_i, y_i$ for which~$u$ acts as a \emph{helper}. 
		Helping such a \emph{large-large} connection, requires six entries in the forwarding table and six ports: three entries and ports for each tree, 
		\emph{ego-}$\BST(x_i)$ and~\emph{ego-}$\BST(y_i)$.
	\end{itemize}

%\item 
\noindent If~$u$ is \emph{large},~$u$'s forwarding table is simpler and contains: 
	\begin{itemize}
	\setlength{\itemindent}{0pt}
		\item A (physical) link to the current root of~\emph{ego-}$\BST(u)$. %(a large node).
		\item A set of virtual roots to improve the performance of~\emph{ego-}$\BST(u)$.\footnote{The use of virtual roots is
a practical optimization. In a traditional self-adjusting
BST, the root changes over time,
depending on the demand: accessed
elements are moved to the root.
In~$\system$, a node~$u$ uses
a set of virtual pointers to implement
the root of ~\emph{ego-}$\BST(u)$: 
the root is implemented using a 
\emph{constant} set of nodes (all at distance 1),
managed in a first-come-first-serve queue, evicting
the least-recently used (lru) root. 
However, this optimization does
not affect the asymptotic performance of our network.} 
		%pointing to low-degree nodes. \chen{fix this, add to appendix}
	\end{itemize}
%\end{itemize}

\subsubsection{Roles}

The algorithms underlying~$\systems$ involve four different node roles: 
\begin{itemize}[leftmargin=9pt]
\item \textbf{The Source} (Algorithm~\ref{alg:sender}):
Let~$u$ be the source of a communication
request~$(u,v)$. In case~$u$ is a  large node,
it will simply forward the request to the root of
\emph{ego-}$\BST(u)$ (or directly to~$v$ if it is one of the virtual roots of~\emph{ego-}$\BST(u)$).
In case~$u$ is a small node then: if~$v \in S(u)$ then it will forward it directly to $v$; else,
if $u$ participates in \emph{ego-}$\BST(v)$ then $u$ will forward it to its parent 
in~\emph{ego-}$\BST(v)$. Else, If~$v$ is a new communication partner,~$u$ will notify
the coordinator and request being connected to~$v$. 
%The behavior is summarized in 
%Algorithm~\ref{alg:sender}.

\item  \textbf{The Forwarder} (Algorithm~\ref{alg:fwd}):
A forwarder~$x$  is a node which is
neither the source nor the destination of
a communication
request~$(u,v)$, i.e.,~$u\neq x \neq v$. It acts as an inner node in 
the ego-tree network and may also be a helper
(\emph{see below}). 
By our construction, node~$x$ must hence either be part of 
\emph{ego-}$\BST(u)$ or~\emph{ego-}$\BST(v)$, \emph{or both} (if it is a helper). 
If~$x$ is part of~\emph{ego-}$\BST(v)$ of the destination, 
it needs to forward the request toward the root~$v$; else if~$x$ is part of  
$\BST(u)$ of the source, it needs to forward the request to the correct child based on the ID of~$v$.  
If~$x$ is a helper, it belongs to both~\emph{ego-}$\BST(u)$ and~\emph{ego-}$\BST(v)$, 
and additionally needs to initiate~\emph{ego-}$\BST(u).\adjust()$ to update the tree.
%Otherwise, it behaves like a helper node.

\item  \textbf{The Destination} (Algorithm~\ref{alg:destination}):
The behavior of the destination~$v$  of 
a given communication
request~$(u,v)$ is 
simple: it delivers the request to the upper layer
and if needed, triggers
an~$\adjust()$ operation on the~\emph{ego-}$\BST(w)$, for which the packet was delivered:
this optimizes the network to account for recent communications.

\item \ \textbf{The Coordinator} 
(Algorithm~\ref{alg:controller}): The coordinator 
keeps track of which nodes are small, which nodes are large, and which small nodes have room in their 
forwarding table to help large-to-large edges.
To serve an \emph{addRoute}$(u \rightarrow v)$ request,
the coordinator distinguishes between different cases, potentially
resetting the forwarding tables (using \emph{reset}$()$,
see below), adding helper nodes where needed
or rendering the source and/or destination node large
(using \emph{makeLarge}$()$, see below).
In the simplest case, both~$u$ and~$v$ are small and
the coordinator can instruct the two nodes to connect directly.
If one node is large and one small, the route request is served
by inserting one node in the other node's~\emph{ego-}$\BST$. 
Only if both nodes are large, the coordinator finds a helper node,
which is used to relay between the two~\emph{ego-}$\BST$s, which must already exist.
  
When the coordinator learns that a node~$u$
needs to become large, it invokes the \emph{makeLarge}$(u)$ method,
which instructs the creation of~\emph{ego-}$\BST(u)$. On this occasion, the coordinator
iterates over the working set of~$u$: in case of a small neighbor~$v$,
$v$ is inserted into the~\emph{ego-}$\BST(u)$ directly; otherwise, a helper node is
used. 

The coordinator also instructs the nodes to reset
their working sets if no more helper nodes are available,
i.e., if the total sizes of the working sets is~$n \theta/2$ and the network is, what we call, \emph{full}.
Concretely, using the \emph{reset}$()$ method, the coordinator instructs
all nodes to clear their forwarding tables (i.e., working sets).

\end{itemize}

\begin{algorithm}[t!]
\caption{Source~$u$, upon request~$u \rightarrow v$}\label{alg:sender}
	\begin{algorithmic}[1]
%		\REQUIRE 
%		\ENSURE todo
%		\STATE~$\wrkset(u)=\wrkset(u)\cup v$
		\IF{$u$ is large}
			\STATE \emph{forward} to root of~\emph{ego-}$\BST(u)$
		\ELSE 
			\STATE (* small node *)
			\IF{$v \in S(u)$}
				\STATE \emph{forward directly} to~$v$
			\ELSIF{\emph{ego-}$\BST(v) \in L(u)$}
				\STATE ~\emph{ego-}$\BST(v).\forward(v)$ (to parent)
			\ELSE %\chen{new node, to to controller. messages to controller should look diffrent}
				\STATE (* new partner *)
				\STATE \textbf{notify coordinator:} \emph{addRoute}$(u \rightarrow v)$		
			\ENDIF
		\ENDIF
	\end{algorithmic}
\end{algorithm}

%\stefan{chen: line 7 of algorithm 1.  ego-BST(v) refers to a set (all nodes in this tree), so the line needs to be fixed.}

\begin{algorithm}[t!]
\caption{Forwarder~$x$, for request~$u \rightarrow v$}\label{alg:fwd}
	\begin{algorithmic}[1]
		\REQUIRE by definition~$x\in \emph{ego-}\BST(u)$ and/or~$x\in \emph{ego-}\BST(v)$
%		\ENSURE todo
		\IF{$x\in \emph{ego-}\BST(v)$}
			\STATE~\emph{ego-}$\BST(v).\forward(v)$ (to parent)
			\IF{$x\in \emph{ego-}\BST(u)$}
				\STATE (*~$x$ is an helper to~$(u,v)$ *) 
				\STATE~\emph{ego-}$\BST(u).\adjust()$
			\ENDIF		
		\ELSE
			\STATE (*~$x\in \emph{ego-}\BST(u)$ *) 
			\IF{$\exists$ child~$x$ toward~$v$}
				\STATE~$\emph{ego-}\BST(u).\forward(v)$ (to child)
			\ELSE			
				\STATE \textbf{notify coordinator}: \emph{addRoute}$(u\rightarrow v)$	
			\ENDIF	
		\ENDIF
	\end{algorithmic}
\end{algorithm}

\begin{algorithm}[t!]
\caption{Destination~$v$, upon request~$u \rightarrow v$}\label{alg:destination}
	\begin{algorithmic}[1]
%		\REQUIRE 
%		\ENSURE todo
		\STATE process packet
		\IF{request received on some~\emph{ego-}$\BST(w)$ }
			\STATE~\emph{ego-}$\BST(w).\adjust()$
		\ENDIF
	\end{algorithmic}
\end{algorithm}

\newcommand\tind{\hspace{-50 pt}}
\newcommand\ifind{\hspace{10 pt}}

\begin{algorithm}[t!]
	\caption{Coordinator}\label{alg:controller}
	\hspace{-5 pt}  \textbf{\emph{addRoute}}$(u \rightarrow v)$:
	\begin{algorithmic}[1]
		\IF{the total size of all working sets is ~$\theta n/2$}
			\STATE (* network is \emph{full} *)  
			\STATE~$\reset()$
		\ENDIF
		\STATE (* available helper~$x$ must exist  *) 
		\STATE add $v$ to the working set of $u$, $\wrkset(u)$
		\IF{$u$ is small but $\abs{\wrkset(u)} = \theta + 1$ } 
%			\IF{$u$ is helper}
%				\STATE replace~$u$ with an new helper
%			\ELSE
%			  	\STATE \textbf{\emph{makeLarge}}$(u)$
%			\ENDIF
			%\STATE \textbf{if}~$u$ is a helper \textbf{find} a replacement~$x$ 
			\STATE \textbf{\emph{makeLarge}}$(u)$
		\ENDIF
		\STATE add $u$ to the working set of $v$, $\wrkset(v)$
		\IF{$v$ is small but $\abs{\wrkset(v)} = \theta + 1$ }
			%\IF{$v$ is helper}
			%	\STATE replace~$v$ with an new helper
			%\ELSE
			% 	\STATE \textbf{\emph{makeLarge}}$(v)$
			%\ENDIF
			%\STATE \textbf{if}~$v$ is a helper find replacement~$x$ 
			\STATE \textbf{\emph{makeLarge}}$(v)$
		\ENDIF
		\STATE (* available room in both tables to add edge $(u \rightarrow v)$ *) 
		\STATE \textbf{cases:}
		\STATE \ifind \textbf{if}~$u$ and~$v$ small:
			%\STATE~$D(u)=D(u)\cup\{v\}$,~$D(v)=D(v)\cup\{u\}$
			\STATE \ifind  \ifind $u$ connects directly to~$v$, update~$S(u), S(v)$
		\STATE \ifind \textbf{if}~$u$ small and~$v$ large:
			\STATE \ifind  \ifind  \emph{ego-}$\BST(v).\insrt(u)$, update~$L(u)$
		\STATE \ifind \textbf{if}~$u$ is large and~$v$ is small:
			\STATE \ifind  \ifind  \emph{ego-}$\BST(u).\insrt(v)$, update~$L(v)$
		\STATE \ifind \textbf{if}~$u$ is large and~$u$ is large:
			\STATE \ifind  \ifind find a helper node~$x$
			\STATE \ifind  \ifind \emph{ego-}$\BST(u).\insrt(x \text{ as } v)$
			\STATE \ifind  \ifind  \emph{ego-}$\BST(v).\insrt(x \text{ as } u)$	
		\end{algorithmic}
	\hspace{-5 pt} \textbf{\emph{makeLarge}}$(u)$:
		\begin{algorithmic}[1]
			\STATE create ~\emph{ego-}$\BST(u)$
			\FOR{each~$v\in \wrkset(u)$}
%				\STATE \emph{addRoute}$(u,v)$
				\IF{$v$ is small}
					\STATE~\emph{ego-}$\BST(u).\insrt(v)$	
				\ENDIF
				\IF{$u$ is large}
					\STATE find a helper node~$x$
					\STATE~\emph{ego-}$\BST(u).\insrt(x \text{ as } v)$
					\STATE~\emph{ego-}$\BST(v).\insrt(x \text{ as } u)$
				\ENDIF
			\ENDFOR		
		\end{algorithmic}
	\hspace{-2 pt}\textbf{\emph{reset}}$()$:
	\begin{algorithmic}[1]
			\FOR{each~$u\in V$}
				\STATE \textbf{inform} to clear~$S(u)$,~$L(u)$,~$\wrkset(u)$
				\STATE \textbf{set}~$u$ to small
			\ENDFOR		
		\end{algorithmic}
\end{algorithm}

\section{Analysis of Static Optimality}\label{sec:analysis}

This section formally proves the properties and performance 
of~$\systems$. To improve readability,
some lemmas and proofs are deferred to the appendix.
We call a communication sequences~$\sigma$ 
\emph{sparse} if it is~$(c,\delta)$-sparse 
for a constant~$c$ and~$\delta = \Omega(c n D)$
(cf~Definition~\ref{def:sparse}),
where~$D$ is the routing cost on the demand-oblivious network 
(henceforth usually assumed to be~$\Theta(\log n)$, the minimum
possible diameter for a scalable, constant-degree network).

We now show that forwarding in~$\systems$ 
 does not require a global routing algorithm and can use arbitrary addressing, and its size is bounded by a constant $\Delta=6\theta=24c$.

\begin{theorem}\label{thm:dynopt}
For any sequence~$\sigma$, 
$\systems$ provide~$\Delta$-compact
and local routing, as well as arbitrary addressing.
\end{theorem}
\begin{proof}
\noindent \emph{Local routing.}
The proof of the local routing property
is by construction and
the states of the source~$u$ and the destination~$v$.
For routing a request from a small node~$u$ to a small node~$v$,
the packet is directly forwarded to the destination~$v$.
For routing a request from a small node~$u$ to a large node,
the packet is forwarded to~$v$ by traversing~\emph{ego-}$\BST(v)$ from parent
to parent, until the root of ~\emph{ego-}$\BST(v)$ is reached,
and from there directly to~$v$.
For routing a request from a large node~$u$ to a small node~$v$,
the packet is forwarded to~$v$ by traversing~\emph{ego-}$\BST(u)$ from the root of 
\emph{ego-}$\BST(u)$ downward to~$v$, similar to a classic search on a binary search tree.
For routing a request from a large node~$u$ to a large node~$v$,
the packet is forwarded to~$v$ in two steps. By construction, 
there must exist a helper node~$x$
that participates both in~\emph{ego-}$\BST(u)$ and~\emph{ego-}$\BST(v)$. 
First the request is forwarded on~\emph{ego-}$\BST(u)$ downward 
to~$x$ (which is stored in the tree as~$v$).
Then,~$x$ notes that the destination is~$v$ and forwards
the packet upward to~$v$, on~\emph{ego-}$\BST(v)$. Since all binary search trees 
in the system are maintained locally using the~$\adjust()$  method, 
no global routing algorithm is needed.

\noindent \emph{Compact routing.}
We set the threshold $\theta$ to be twice the largest possible 
average degree in a window of size at most $\delta$, i.e.,~$\theta = 4c$, so every node with
working set size less than~$\theta$ is \emph{small}, and otherwise, 
\emph{large}.
Let ~$\Delta=6\theta$ (a constant) 
and we set the forwarding table to size 
$\Delta$, so a~$\system$ supports compact routing.
We need to show that the forwarding table does not exceed the size $\Delta$.
As long as the coordinator did not call~$\reset()$,
for a \emph{large} node the forwarding table is by design at most $\Delta$, it contains one link to its \emph{ego-}$\BST$ root and a set of links to at most $\Delta-1$ virtual roots.
For a \emph{small} node we prove this in Lemma \ref{lem:helping} that we state later.

\noindent \emph{Arbitrary addressing.} The 
support for arbitrary addressing follow by design, since the search operation in binary search trees can support it naturally. 
\end{proof}

We can now state our main result:
\begin{theorem}\label{thm:staticopt}
For any $(c,\delta)$-sparse communication sequence~$\sigma$, 
where~$\card{\sigma} \ge \delta = \Omega(nD)$,
there is a constant $\Delta$ for which
$\systems$ relying on statically optimal ego-BSTs 
(e.g.~splay trees~\cite{splaytrees}), 
are statically optimal for~$\Delta$-degree bounded networks.
%where~$\Delta$ is a constant.
\end{theorem}
\begin{proof}
We again set the threshold to be twice the 
average degree~$\theta = 4c$, 
let ~$\Delta=6\theta$ (a constant), 
and set the forwarding table to size $\Delta$.
Let~$\netw^*$ be the optimal~$\Delta$-degree bounded 
network used by the optimal static algorithm~$\STAT(\sigma)$. 
From~\cite{disc17} it follows that the average cost of 
$\STAT$ is lower bounded by: 
\begin{eqnarray}
\Cost(\STAT,\netw^*, \sigma) \ge \Omega \left(\max(H_\Delta(\hat{Y}_\sigma \mid \hat{X}_\sigma), H_\Delta(\hat{X}_\sigma \mid \hat{Y}_\sigma)) \right)
\end{eqnarray}

%\noindent (for completeness, at the end of this section,
%we show how to adapt the corresponding
%result).

We will prove static optimality of~$\systems$ in two steps.
First, we will show that the routing cost of a~$\system$ is optimal and 
proportional to its trees adjusting cost. Second, we will bound the cost of 
the operations and messages that are related to the coordinator in the~$\system$. 
Overall, we will show that the amortized cost of a~$\system$ is 

\begin{eqnarray}
\Cost(\system, \emptyset, \sigma) \le O\left(H(\hat{Y}_\sigma \mid \hat{X}_\sigma) + H(\hat{X}_\sigma \mid \hat{Y}_\sigma)\right)
\end{eqnarray}
\noindent making it order optimal since~$\Delta$ is constant
(recall that~$\netw_0=\emptyset$ is an empty initial network).

We divide~$\sigma$ into subsequences, $\sigma^{(i)}$, separated by the $i$'th call to the~$\reset()$ 
operations announced by the coordinator.
If no $\reset()$  is called then $\sigma^{(1)}=\sigma$.
If $\reset()$  was called $k$ times then the last (partial) subsequence is denoted by $\sigma^{(k+1)}$.
We start with the analysis of a single ``window'',~$\sigma^{(1)}$, 
which is the subsequence of~$\sigma$ from the start until the first~$\reset()$ 
operation.
The length of~$\sigma^{(1)}$ is~$\Omega(c n D)$ by assumption on sparsity and it contains exactly 
$c n$ unique requests (see the Coordinator algorithm, Algorithm~\ref{alg:controller}).
We claim the following:
\begin{claim}
$\system$ is statically optimal on~$\sigma^{(1)}$.
\end{claim}
\begin{proof}
Let~$H^{(1)}_{\mathrm{con}} = \max(H(\hat{Y}_{\sigma^{(1)}} \mid \hat{X}_{\sigma^{(1)}}), H(\hat{X}_{\sigma^{(1)}} \mid \hat{Y}_{\sigma^{(1)}}))$ 
be the maximum of the conditional entropies.
We need to show that~$\system$'s cost on~$\sigma^{(1)}$ is~$O(H^{(1)}_{\mathrm{con}})$ 
to prove its optimallity. We separate the cost into two groups:

\noindent \emph{\bf Routing and BST adjustment cost}: 
For analytical reasons, we consider a symmetric 
version of~$\sigma^{(1)}$, named~$\bar{\sigma}^{(1)}$,
which keeps the total number of requests between pairs the same
but divides them half-half:
for each (directed) request~$(u,v)$, we consider that half of the requests
went the other direction,~$(v, u)$, making the number of requests
for a pair~$u,v$ equal in both directions. This makes the frequency 
matrix (a matrix representation of the pair's frequencies in the demand) 
of ~$\bar{\sigma}^{(1)}$ symmetric.  Theorem~\ref{thm:symmetirc}, that we prove later, 
states that~$H(\hat{Y}_{\bar{\sigma}^{(1)}} \mid \hat{X}_{\bar{\sigma}^{(1)}}) 
= H(\hat{X}_{\bar{\sigma}^{(1)}} \mid \hat{Y}_{\bar{\sigma}^{(1)}}) \le H^{(1)}_{\mathrm{con}} +1$.

We consider the cost by node type. For a \emph{small} 
node, if it connects to another \emph{small} node,
then the two nodes have a direct connection that starts from 
the first request and stays active for the whole~$\sigma^{(1)}$
(unless the node becomes \emph{large}, which we address later). 
The amortized cost for such a request is one. If a \emph{small node} 
connects to a \emph{large} node, we charge the cost for routing and 
the cost for adjusting the network
to the large node, which we discuss now.
Each \emph{large} node~$w$ maintains a~\emph{ego-}$\BST(w)$ for its 
communication partners. 
Since~\emph{ego-}$\BST(w)$ is assumed to be a statically optimal datastructure~\cite{splaytrees}
 on all requests for which~$w$ is the source or destination (recall that since~$\bar{\sigma}^{(1)}$ is symmetric, the frequency distribution of destinations
 \emph{from}~$w$,~$Y_w$, and sources \emph{to}~$w$,~$X_w$, are the same),
 it follows that the cost of these requests (routing plus adjustments) 
 is~$O(H(Y_w)) = O(H(X_w))$ (see Lemma~\ref{lem:treeentropy}). 
 This cost includes all~\emph{ego-}$\BST(w).\forward()$,~\emph{ego-}$\BST(w).\insrt()$ 
 and ~\emph{ego-}$\BST(w).\adjust()$ operations. 
Since each routing request involves at most one
forwarding operation by a helping node between two trees (for large-large edges),
the (amortized) cost of routing and tree adjustment is at most 
$H(\hat{Y}_{\bar{\sigma}^{(1)}} \mid \hat{X}_{\bar{\sigma}^{(1)}}) + H(\hat{X}_{\bar{\sigma}^{(1)}} \mid \hat{Y}_{\bar{\sigma}^{(1)}})$, as required.

\noindent \emph{\bf Coordinator messages cost}: 
We discuss the coordinator functions one-by-one:  
\begin{description}
\setlength{\itemindent}{-10pt}
\item[\textbf{\emph{reset()}}:]  Happens once during 
the window~$\sigma^{(1)}$.
The cost is~$n$, to broadcast the reset message to all 
nodes on the fixed network (using a broadcast tree).
 \item[\textbf{\emph{makeLarge()}}:] Happens at most
 once to each node during the window. When \emph{makeLarge}
 is executed at node~$u$, we are guaranteed to have enough helper 
 nodes if needed (since the network is not full yet, see Lemma~\ref{lem:helping}),
 and since we do not add new edges to $u$, we  only replace a constant number of existing
 edges~$(u,v)$  (direct or via tree~\emph{ego-}$\BST(v)$),
 with a new connection via the newly created \emph{ego-}$\BST(u)$ of constant size.
 In each call of \emph{makeLarge},~\emph{ego-}$\BST(u).\insrt()$ is
  amortized (accounting for in the adjustment cost above). The only additional 
  cost is to notify helpers, but the number of helpers is bounded by~$cn$ and sending a message is at most~$O(D)$, so the total cost is~$O(cnD)$.
  \item[\textbf{\emph{addRoute()}}:] Happens exactly~$cn$ times 
  during the window~$\sigma^{(1)}$. The cost of~\emph{ego-}$\BST(u).\insrt()$ 
  and/or~\emph{ego-}$\BST(v).\insrt()$ are  amortized. The only cost is to
  notify the helper node which is at most~$O(D)$. 
  The cost during the window is therefore~$O(cnD)$.
\end{description}
Summing up the total cost of the coordinator messages gives~$O(cnD)$. Since the 
number of requests in the window is $\delta = \Omega(c n D)$, the amorized
cost per coordinator request is~$O(c')$, for a constant~$c'$.
To this we need to add for each \emph{ego-}$\BST(w)$ its amortized cost for routing and adjusting, but this as mention is proportional to 
$O(H(Y_w)) = O(H(X_w))$. Therefore the total amorized cost 
for the window (including routing, adjusting and coordinator messages) is 
$O(H^{(1)}_{\mathrm{con}})$.
\end{proof}

To conclude the proof of  Theorem~\ref{thm:staticopt},
we divide~$\sigma$ into subsequences~$\sigma^{(i)}$,
 separated by~$\reset()$ operations.
 A lower bound for~$\STAT(\sigma)$ is:
~$$\Cost(\STAT,N^*, \sigma) \ge \sum_i \Omega (H^{(i)}_{\mathrm{con}})$$
 While the cost for~$\system$ is:
~$$\Cost(\system,\emptyset, {\sigma}) \le \sum_i O (H^{(i)}_{\mathrm{con}})$$
\noindent which makes~$\system$ statically optimal.
\end{proof}
Note that the coordination cost of the last subsequence (which may be shorter than $\delta$), can be  amortized by  
coordination cost of $\sigma^{(1)}$ so its amortized cost is also a constant.

Observe that the cost of~$\system$ could be much lower than 
the cost of~$\STAT$, since~$\STAT$ is also lower bounded by the conditional entropy 
of the whole demand~$\sigma$, and not only the sum of entropies of the windows.   
% The potentially large difference was also highlighted empirically
%in Figure~\ref{fig:facebook} on the Facebook trace.
% More formally, we can prove the following:
 
\begin{theorem}
The amortized cost of~$\system$ can be up to~$\log n$ times lower than 
the cost of ~$\STAT$.
\end{theorem}
\begin{proof}
Consider for example a demand~$\sigma$ that is 
the concatenation of~$n$ demand
subsequences~$\sigma^{(1)},\sigma^{(2)},\ldots, \sigma^{(n)}$. Each demand 
$\sigma^{(i)}$ is of length~$\Theta(n \log n)$, 
is sparse and has a demand graph which is a (different) two-dimensional grid.
% as in 
%Figure~\ref{fig:example} (e). 
Therefore the amortized cost of~$\system$ for each 
$\sigma^{(i)}$ is \emph{constant}. But if the~$\sigma^{(i)}$ is \emph{different} each time
(e.g., selected round robin),
then~$\sigma$ could be made to be uniform, where overall each source communicates
to all destinations with equal frequency (over the entire~$\sigma$). This will force a lower bound on~$\Omega(\log n)$, the entropy of the uniform distribution, for~$\STAT$,.
\end{proof}

%\begin{theorem}\label{thm:workingset}
%$\systems$ feature the working set property.
%\end{theorem}
%\begin{proof}
%fixme
%\end{proof}
%
%\begin{theorem}\label{thm:dynopt}
%$\systems$ inherit dynamic optimality.
%\end{theorem}
%\begin{proof}
%fixme
%\end{proof}

%We defer the remainder of the proof to the Appendix due to space constraints.

\section{Related Work}\label{sec:relwork} 

%\stefan{chen please read this section}

Datacenter networks have become a critical 
infrastructure and especially the popularity of
online services~\cite{survey2017datacenter}
(e.g., web search, social networks, storage, financial services,
multimedia, etc.) 
has led to a fast increase of 
datacenter traffic~\cite{singh2015jupiter,cisco-pro}.
Many applications (such as scatter-gather
and batch computing applications) generate much
\emph{internal} datacenter traffic,
and consequently, the traffic staying
inside the datacenter is often much larger than
the traffic entering or leaving the datacenter~\cite{survey2017datacenter}.
It is hence not surprising that the design of effective 
and efficient (also in terms of cost and cabling)
datacenter networks  
has received much interest over the last years~\cite{clos,dcell,xpander,bcube,jellyfish,mdcube,vl2,f10,fat-free,besta2014slim,Longhop}.
The situation has recently been compared to 
the early 1980s, when many new interconnection
network designs were proposed~\cite{akella2015universal},
not for datacenters, but for parallel computers.
%In fact, researchers have found that modern datacenter 
%networks also follow (consciously
%or not) some 
%of the fundamental principles guiding the design of ``networks
%on a chip'', e.g., 
%Leiserson's cost universality~\cite{leiserson}:
%the idea that a good network design can emulate
%any other network that could be built at the same cost with
%limited slowdown. 

The advent of technologies for reconfigurable 
(a.k.a.~malleable~\cite{singla2010proteus}) 
networks introduces an additional degree of freedom
to the datacenter network design problem~\cite{singla2010proteus,projector,Jia2017,reactor, helios, firefly, zhou2012mirror,augmenting,chen2014osa}.
By relying on movable antennas~\cite{augmenting},
mirrors~\cite{firefly,zhou2012mirror}, and ``disco-balls''~\cite{projector},
novel technologies in the context of optical circuit switches~\cite{helios,rotornet,reactor,mordia},
60 GHz wireless communication~\cite{zhou2012mirror,kandula2009flyways}, 
free-space optics~\cite{firefly,projector},  
provide unprecedented topological flexibilities, allowing to 
adapt the topology to traffic demands.

Indeed, the physical topology is the next frontier 
in an ongoing effort to render communication networks more
flexible and reconfigurable. Over the last years, 
reconfigurable technologies (typically software) already
enabled various innovations in important domains such as 
traffic engineering~\cite{swan,b4}, load-balancing~\cite{ananta,maglev}, 
and switching~\cite{p4,smartnic}.

While the discussions on the benefits and limitations of such technologies
are still ongoing~\cite{singla2010proteus},
the community has identified a number of benefits of more flexible networks. 
While full bisection bandwidth allows for a
flexible placement and scale-out of applications across clusters~\cite{singh2015jupiter,vl2,clos,Roy:2015},
the cost of (and energy consumed in) traditional 
networks is high, and 
today many datacenters are over-subscribed 
(e.g., fat-trees or folded Clos networks
use a subset of possible roots only).
Some empirical evaluations show that for certain workloads,
a demand-aware network 
can achieve a performance similar to a demand-oblivious
full-bisection bandwidth network at 25-40\% lower cost~\cite{firefly,projector}. %\chen{important move earlier}
Empirical studies also confirm that communication patterns 
are often \emph{sparse} and of \emph{low entropy}, which can be exploited in demand-aware networks:
in~\cite{projector}, it is shown that a high percentage of rack pairs
does not exchange any traffic at all, while less than 1\% of them account
for 80\% of the total traffic. 
In general, 
most bytes
are delivered by large flows~\cite{vl2,judd2015attaining,alizadeh2010data}.
Furthermore, the need for reconfigurability is motivated
by empirical studies showing the difficulty of estimating traffic matrices 
and predicting the future~\cite{kandula2009nature,cao2016joint,hu2015coarse}. 

We also note that the study of reconfigurable networks is not limited
to datacenter networks. Interesting use cases also arise
in the context of wide-area networks~\cite{Jia2017,jin2016optimizing}
and, more traditionally, in the context of overlays~\cite{scheideler2009distributed,ton15splay,ratnasamy2002topologically}.

Existing self-adjusting demand-aware network designs
typically rely on some estimate or snapshot of the traffic demands, 
from which an optimized network topology is (re)computed
periodically
(often using exact algorithms or heuristics)~\cite{ancs18,singla2010proteus,fat-free}.
In contrast, inspired by self-adjusting datastructures, 
we in this paper present a more refined model, 
accounting also for the reconfiguration
costs, and allowing us to study (within our model) 
the tradeoff between the 
costs and benefits of reconfigurations.
We presented a survey and taxonomy of the problem space 
in~\cite{ccr-dan}, and adopt the terminology
introduced there.
%To comply with the double-blind 
%guidelines, we note that 
%In the current paper we adopt the terminology
%of~\cite{ccr-dan},
However, in contrast to the current paper, 
no technical results were derived in the survey paper.
Moreover, in contrast to existing algorithms relying
on mixed integer programming, our algorithms
are efficient (\emph{polynomial-time}), and in contrast to
existing algorithms relying on heuristics, our approach comes
with \emph{provable} guarantees, even over time.
Moreover, 
existing solutions to provide on-demand
bandwidth between communicating servers or top-of-the-rack
switches are sometimes limited to (or at least prioritize
in their opportunistic network~\cite{projector}) 
\emph{direct} connections~\cite{fat-free}. 
We believe that local routing (as supported by~$\systems$) 
is an intriguing property that can simplify 
multihop routing, by
supporting topological adjustments,
without the need for global route recomputations.

In terms of formal guarantees, 
an upper bound on what can be achieved in terms of 
statically optimized demand-aware networks
is due to Avin et al.~\cite{disc17}, who build 
upon initial insights in~\cite{ton15splay,infocom19splay}. 
We in this paper leverage 
the degree reduction technique of~\cite{disc17},
%and the local routing technique~\cite{ton15splay},
however, to derive a very different result.
The fixed demand-aware network designs by
Avin et al.~\cite{disc17} have recently also
been extended to optimize for load, in addition to
route lengths~\cite{infocom19dan}. 

SplayNets~\cite{ton15splay,infocom19splay} also
rely on splay trees to adjust the network, 
and dynamically adapt to changing traffic patterns. However,
besides their convergence properties under specific
fixed demands,
these networks do not provide any
optimality but only heuristic guarantees. In fact, as we discussed, this is an inherent limitation, 
as static optimality is impossible to achieve based on single tree networks.
Indeed, to the best of our knowledge, 
so far, no result existed on how to actually match 
the lower bound provided in~\cite{ton15splay},
without perfect knowledge of the demand. 

%Finally, 
We also note that our approach of reconfiguring
network topologies to reduce communication
costs, is orthogonal to optimization approaches changing the
traffic matrix itself (e.g.,~\cite{fibium}) or 
migrating communication endpoints on a fixed topology~\cite{disc16}.

%\stefan{add p2p literature!}

\section{Conclusion}\label{sec:conclusion}

This paper presented the first 
self-adjusting network which provides 
entropy-proportional
(and hence statically optimal) route lengths
and reconfigurations,
constant-sized forwarding tables and local routing.
Our approach leveraged an intriguing connection to 
self-adjusting datastructures, leveraging self-adjusting
BSTs as building blocks (i.e., per-source ``ego-trees''), 
and combining them to a \emph{network}.

We believe that our work opens several interesting
directions for future research. 
In particular, an intriguing open question from our work  
regards the design of self-adjusting DANs which
optimize metrics related to \emph{temporal}
locality, such as working sets, or even achieve 
dynamic optimality in specific settings.
More generally, we believe that self-adjusting
networks can be of interest beyond the datacenter
context considered here; for example, it may be interesting
to explore self-adjusting peer-to-peer overlays.

%\stefan{opportunity for p2p community!}

{\balance
\footnotesize 
\renewcommand{\baselinestretch}{.95}

}

%\begin{comment}
\appendix

\section{Intuition and Examples}\label{sec:motivation}

This section establishes the connection between 
entropy of demand and route lengths in fixed
demand-aware networks, which lies
at the heart of the static optimality of the
self-adjusting DANs presented in this paper. 
While a formal proof appeared in~\cite{disc17},
the objective of this section is to provide intuition
and examples, as well as first empirical results.

In the following, the route length to serve request~$\sigma_t=(u,v)$,
is given by the hop distance~$\d_{\netw_t}(u,v)$ 
from~$u$ to~$v$, along the routing path chosen
by the algorithm over~$\netw_t$. If not specified otherwise we assume that routes are along shortest paths.
Furthermore, 
we will sometimes consider certain time intervals of the request sequence,
and we will denote the requests occurring at any time~$t$ where~$t_1\leq t \leq t_2$ for
two times~$t_1 \leq t_2$ by~$\sigma[t_1,t_2]$. (Hence~$\sigma=\sigma[0,\infty]$.) Moreover, 
we can again think of a subsequence 
$\sigma'\subseteq \sigma$,
as a directed and weighted \emph{demand graph} (or guest graph)~$G(\sigma')=(V(\sigma'),E(\sigma'))$. 
To give an example, Figure~\ref{fig:example2} \emph{(a)} and 
\emph{(b)}  
show two different demand graphs,
each for a different time interval, $t_1$ and $t_2$ receptively. Since the communication traffic 
is changing over time, the demand graph also changes 
(the edges and their weights). 

\subsection{Changing Demands Require Reconfigurations}

If the communication patterns are sufficiently
different, it can make sense to reconfigure the 
adaptive demand-aware network,~$\netw_t$, 
as well. In general and ideally, each ~$\netw_t$ should be optimized 
toward the request graph \emph{of the future}. In this example, 
for simplicity, we set the maximum allowable degree of
$\netw_t$ to be two (i.e., we only have cycles and lines). 
Note that the network should also support multihop routing,
For example in Figure~\ref{fig:example2} \emph{(b)}, node 1 communicates 
with three partners  but its degree is bounded by two,
so it must also know where to forward messages toward node 4, for example.

\begin{figure}[t]
\centering
\begin{tabular}{cc}
\includegraphics[height=.7\columnwidth]{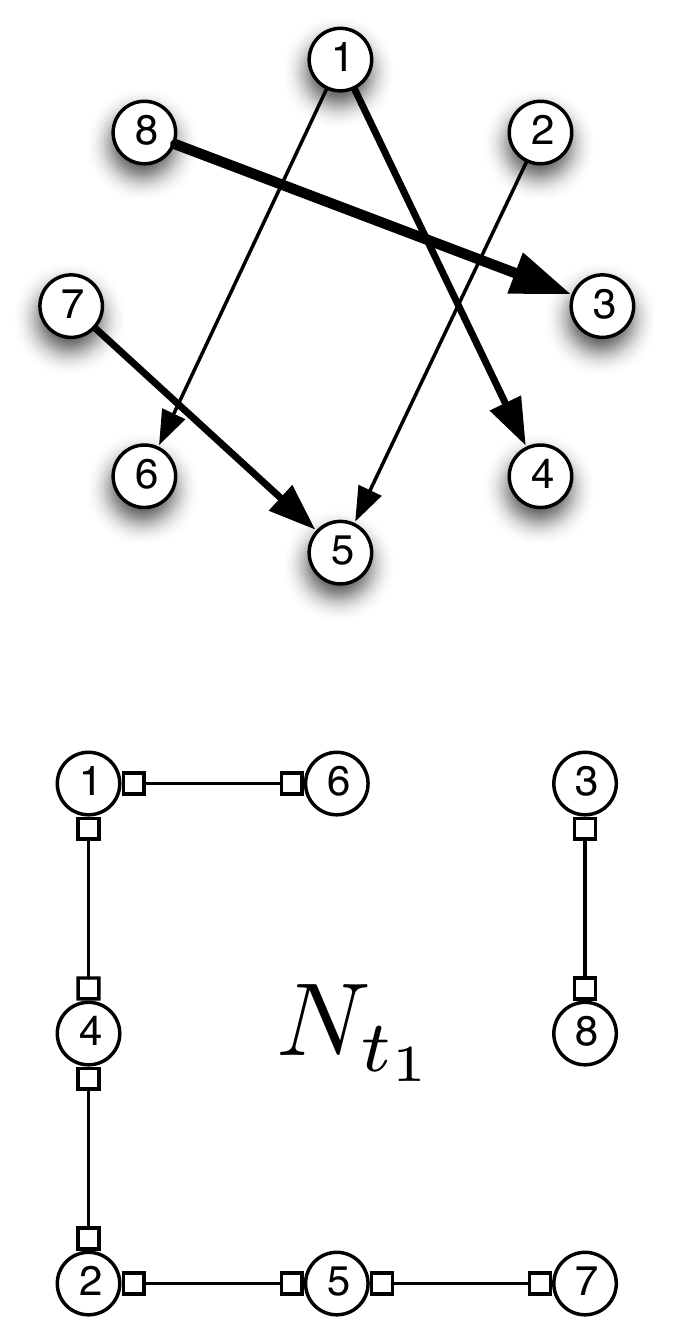} &
\includegraphics[height=.7\columnwidth]{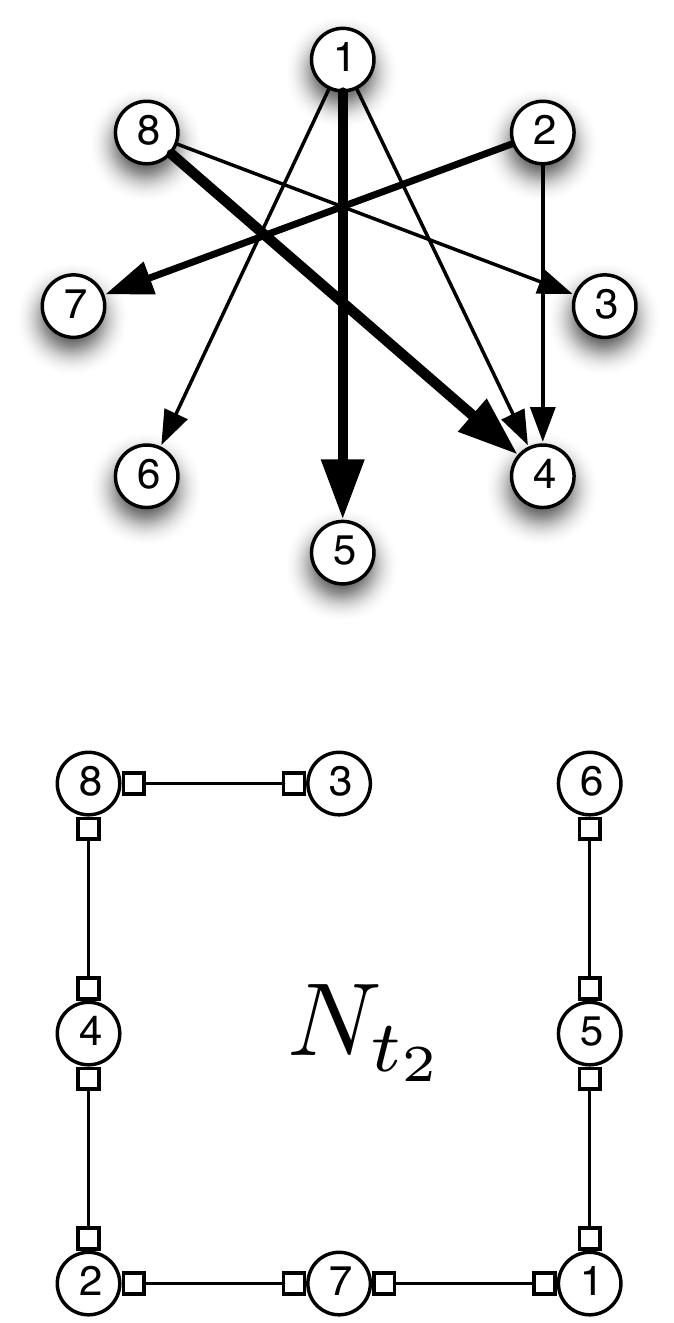} \\
(a) & (b) 
\end{tabular}
\caption{\textbf{The need for reconfigurations.}
 In the upper part, 
two (directed, weighted) demand graphs are shown for two different times: 
(a)~$G(\sigma[t'_1,t_1])$
 and (b)~$G(\sigma[t'_2,t_2])$.
In the lower part, their corresponding demand-aware networks  
~$\netw_{t_1}$ and~$\netw_{t_2}$ of bounded degree two are shown.
(a) In~$\netw_{t_1}$ every node is a direct neigbhor of its communication partner.
(b) In~$\netw_{t_2}$ multihop routing and changing forwarding tables are needed.}
\label{fig:example2}
\end{figure}

\subsection{Demand-Oblivious vs Demand-Aware Networks}

\begin{figure}[t]
\centering
\begin{tabular}{cc}
\includegraphics[height=.65\columnwidth]{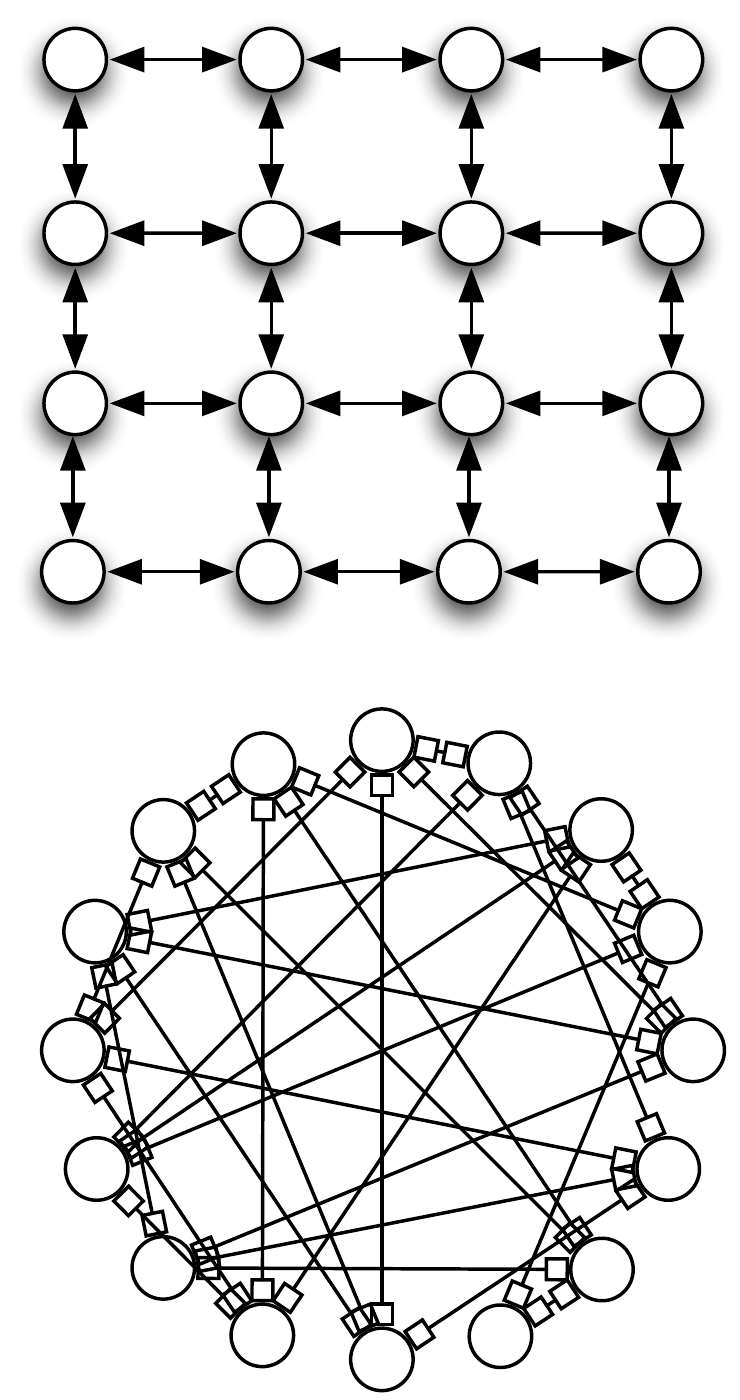} &
\includegraphics[height=.65\columnwidth]{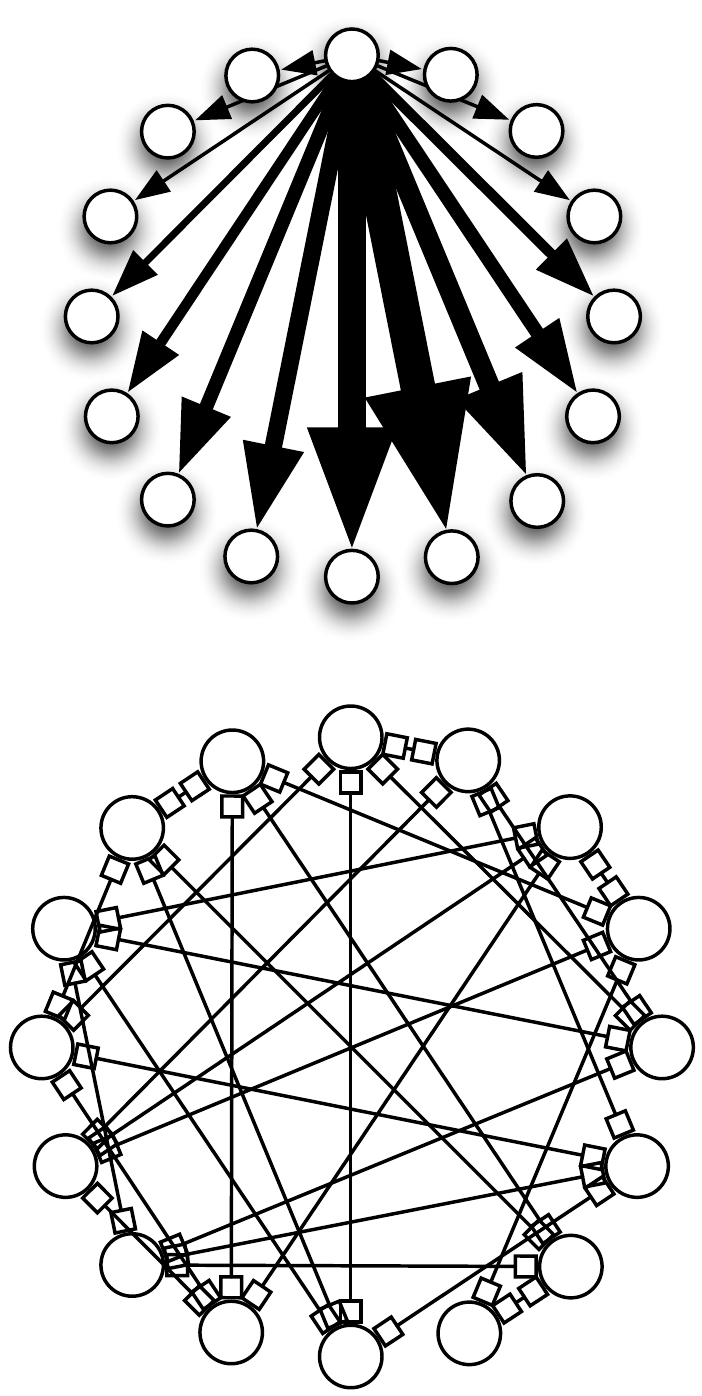} \\
(a) & (b) 
\end{tabular}
\caption{\textbf{Expander networks do not achieve optimal average 
route lengths for sparse demand graphs.}
(a) Oblivious embedding of a 2-dimensional grid demand graph (upper graph) on a constant degree expander network (lower graph) will result in average
route lengths of
$\Omega(\log n)$, while the conditional entropy of the demand graph 
is less than two.
(b) Oblivious embedding of a weighted star demand graph 
on a constant degree expander network will result in 
an average route length of~$\Omega(\log n)$ while the conditional 
entropy of the demand graph could be much lower.}
\label{fig:example3}
\end{figure}

First, we note that the route lengths in demand-oblivious networks 
(such as state-of-the-art expander networks~\cite{hoory,xpander}) cannot be proportional
to the conditional entropy
 and hence cannot provide
the desired solution
(even in the presence of traffic engineering flexibilities~\cite{fat-free}). 
To give a simple example (for the sake of simplicity and 
clarity we leave some of the details out),
consider a workload describing a communication
pattern~$\sigma$ whose demand graph~$\demg(\sigma)$
forms a two-dimensional square grid, of size
$\sqrt{n}\times \sqrt{n}$, see Figure~\ref{fig:example3}~\emph{(a)}.
For this sequence~$\sigma$,~$H(\hat{X})$ and  
$H(\hat{Y})$ are of order~$\log n$, since the frequency of sources and destination is uniform which results in \emph{the maximum possible entropy}.
Embedding this workload 
on a static expander 
in an demand-oblivious way  (i.e., random, or arbitrary)
will result in an 
average route length also in the order of~$\log n$, which is
the diameter of a bounded degree expander. 
However, 
since every node has at most four neighbors, the conditional entropy 
(both~$H(\hat{Y}|\hat{X})$ and~$H(\hat{X}|\hat{Y}))$ is 
only \emph{2}: a gap of~$\Theta(\log{n})$.  A demand-aware design could achieve this bound. 

Another example introducing a large gap of~$\Theta(\log{n})$
between demand-oblivious and demand-aware networks 
%conditional entropy and the entropy of a sequence~$\sigma$
is a demand graph~$\demg(\sigma)$ which forms a star 
(with \emph{unbounded} degree), see Figure~\ref{fig:example3}~(b): 
node pairs communicate at different frequencies (skewed distribution,
as indicated by the thickness). For this demand, the conditional entropy could be much lower than
% the entropy which is about 
$\log n$ which will be the cost of serving this demand on an demand-oblivious expander.
%(the conditional entropy  could even be constant). 

More generally, one can see that 
every sparse communication pattern
%(as it is often observed in practical 
%workloads~\cite{projector}), 
which is
embedded on a demand-oblivious expander,
will result in average route lengths in the order of~$\Omega(\log{n})$, 
the diameter,
regardless of the entropy or the conditional entropy of the demand. 

In order to illustrate the potential gap between the upper bound of entropy (e.g.,$H(\hat{X})$) and lower
bounds of conditional entropy (i.e., $H(\hat{X}|\hat{Y}))$ with some concrete numbers, we plot in Figure~\ref{fig:facebook} the empirical
entropy as well as the conditional empirical
entropy of~$3M$ routing requests
from a Facebook datacenter trace~\cite{roy2015inside}.
%see Figure~\ref{fig:facebook}.
The demand~$\sigma$ consists
of~$n=13,748$ communication partners, note that $\log n =13.74$ for this case
(we consider the binary logarithm).
%In this plot, in Figure~\ref{fig:facebook} (a),
%we consider
 The figure considers 
 times~$t$ that are multiplicatives of~$100K$.
 For each time 
~$t$, the measures are presented both for the full range 
~$\sigma[1,t]$ (labelled~$H$) as well as for a time window of 
 the last~$100K$ 
 requests~$\sigma[t-100K,t]$ (labelled~$H_W$), 
 to shed light on the \emph{temporal locality}.
 Clearly the conditional 
 entropy is lower than the entropy, and in particular the 
 conditional entropy of the window is much less than 
 the entropy of~$X$ and~$Y$. % in the window. 
%In Figure~\ref{fig:facebook} (b),
%for 
% each node~$z$, we show the relationship between the 
% entropy of incoming requests~$H(\hat{X} \vert Y=z)$ 
% and the entropy of outgoing requests~$H(\hat{Y} \vert X=z)$. 
%% There exist nodes which are only senders
%% and nodes which are only receivers, that only send information and nodes 
%% that only receive information, but 
%The general trend is 
% that high entropy source nodes are also high entropy destination nodes.
This indicates that demand-aware designs 
could reduce the average route length in the network.

%\subsection{The Quest for Entropy-\\Proportional Routing}

\subsection{Limitations of Tree Networks}

%\stefan{The core of the problem is, that, in order for the network to be
%optimal, we either need to know (or assume) the future, 
%or, and this is the surprising part, let the network self-adjust, 
%without knowing the future, but accounting for reconfiguration costs.}

In order to be statically optimal,
a \emph{self-adjusting} network hence needs to achieve the conditional entropy bound,
\emph{without knowledge of~$\sigma$},
but using reconfigurations (which come at a \emph{cost}), 
in an \emph{online} manner. 
Note that while the lower bound only holds for \emph{fixed} topology networks, 
a self-adjusting network can in principle perform much better.

The best upper bound known so far  for \emph{self-adjusting} networks
is~$O(H(\hat{X}) + H(\hat{Y}))$, 
where~$H(\hat{X})$ and~$H(\hat{Y})$ are the
%\emph{non-conditional}, i.e.,\emph{marginal} 
empirical entropies of sources and 
destinations in~$\sigma$, respectively~\cite{ton15splay}. 
It is achieved by a self-adjusting \emph{tree} 
network. While this is optimal for some distributions, 
in particular for \emph{product distributions}, in general, 
it is far from optimal.
% as discussed above.

Moreover, it is important to note that, 
following the results in~\cite{ton15splay}, if such an almost optimal 
self-adjusting network exists, then it cannot be a bounded degree \emph{tree}. 
The design we present in this work is far from a tree, 
in fact it is based on a network which is a union of trees.
% and thus does not look like a single tree.

%For example, a bounded degree tree will not be able to serve 
%optimally the two-dimentional grid example mentioned above. 

\begin{figure}[t]
\centering
%\begin{tabular}{cc}
\includegraphics[width=.6\columnwidth]{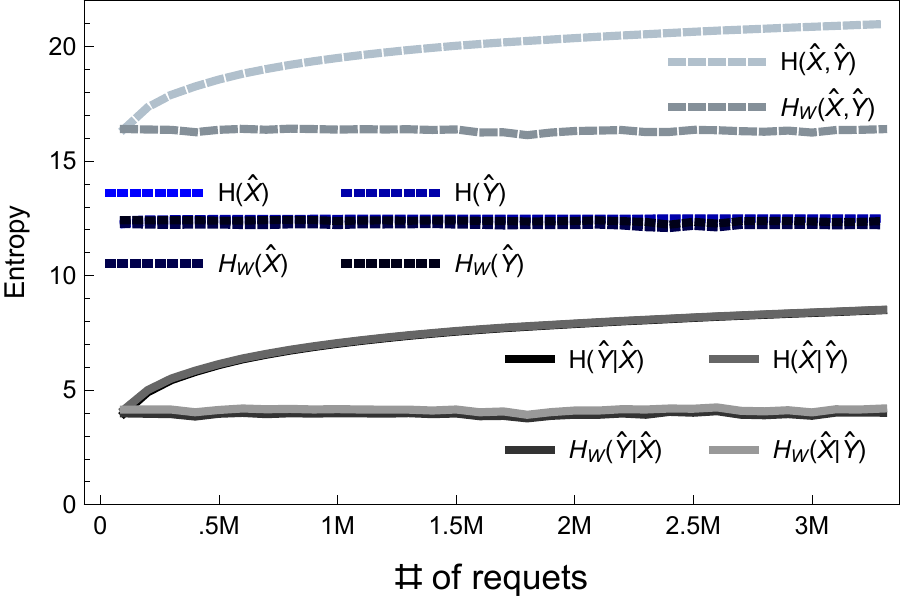}
%\includegraphics[width=1\columnwidth]{figures/fbEntropy2.pdf} 
%\includegraphics[width=.9\columnwidth]{figures/fbSrcDstH.pdf} \\
%\includegraphics[height=.18\textwidth]{figures/fbSrcDstH.pdf} \\
%(a) & (b) %& (c)
%\end{tabular}
\caption{Entropy measures  %vs conditional entropy
 in Facebook's workload.} % traces}
\label{fig:facebook}
\end{figure}

\section{Deferred Proofs}

\subsection{Deferred Analysis of Symmetric Matrix}

Let~$M=M(\sigma)$ be a joint (non-symmetric) frequency matrix 
resulting from~$\sigma$. 
Let~$H_M(Y \mid X)$ denote the conditional 
entropy of~$Y$ given~$X$ under the joint probability distribution~$M$. 
By definition,~$H_M(Y \mid X) = H(\hat{Y}_\sigma \mid \hat{X}_\sigma)$.
Let~$H^*_{\mathrm{con}}=\max(H(\hat{Y}_\sigma \mid \hat{X}_\sigma), 
H(\hat{X}_\sigma \mid \hat{Y}_\sigma))$, the maximum of both possible 
conditional entropies. Let~$\bar{M} = (M+M^T)/2$ be the  
symmetric version of~$M$. The conditional entropies of the symmetric and non-symmetric distributions 
are related as stated in the following theorem: 

\begin{theorem}\label{thm:symmetirc} 
The conditional entropy of the symmetric matrix~$\bar{M}$ cannot be much larger than the maximal conditional entropy of~$M$.
\begin{align}
H_{\bar{M}} (Y \mid X) = H_{\bar{M}} (X \mid Y) \le H^*_{\mathrm{con}} +1
\end{align}
\end{theorem}

%To prove the theorem, we first show the following lemma.

The proof of the theorem mainly follows from the follwoing Lemma that is based on the concavity of entropy~\cite{cover2012elements}
and simple entropies algebra.

\begin{lemma}\label{clm:average}
Let~$\vec{p}$ and~$\vec{q}$ be two probability (frequency) 
distributions for the same set. Let~$H^* = \max(H(\vec{p}), H(\vec{q}))$. Then
\begin{align}
\frac{1}{2} H^* \le \frac{1}{2}H(\vec{p}) + \frac{1}{2}H(\vec{q}) \le H(\frac{\vec{p}+\vec{q}}{2}) \le H^* + 1
\end{align}
\end{lemma}

%\begin{proof}
%The lower bound is implied by 
%the concavity of entropy~\cite{cover2012elements}. i.e.,~$H(\frac{1}{2}\vec{p}+\frac{1}{2}\vec{q}) \ge \frac{1}{2}H(\vec{p}) + \frac{1}{2}H(\vec{q})$.
%For the upper bound, we have:
%\begin{align*}
%H(\frac{\vec{p}+\vec{q}}{2}) &= \sum \frac{p_i+q_i}{2} \log \frac{2}{p_i+q_i} \\
%	&= \frac{1}{2} \sum p_i \log \frac{2}{p_i+q_i} + \frac{1}{2} \sum q_i \log \frac{2}{p_i+q_i} \\
%	& \le \frac{1}{2} \sum p_i \log (\frac{1}{p_i}) + \frac{1}{2} + \frac{1}{2} \sum 
%	q_i \log (\frac{1}{q_i}) + \frac{1}{2} \\
%	&= \frac{1}{2}H(\vec{p}) + \frac{1}{2}H(\vec{q}) +1 \le H^* + 1
%\end{align*}
%\end{proof}

%We now prove Theorem~\ref{thm:symmetirc}.
%\begin{proof}[Proof of Theorem~\ref{thm:symmetirc}]
%From Claim~\ref{clm:average} we have
%$H_{\bar{M}}(X,Y) \le H_{{M}}(X,Y) +1$ and~$H_{\bar{M}}(X) \ge \frac{1}{2}H_{{M}}(X) + \frac{1}{2}H_{{M}}(Y)$.  Now we can bound the conditional entropy.
%\begin{align*}
%H_{\bar{M}} (Y \mid X) &= H_{\bar{M}}(X,Y) - H_{\bar{M}}(X) \\
%				&\le  H_{{M}}(X,Y) +1 - \frac{1}{2}H_{{M}}(X) - \frac{1}{2}H_{{M}}(Y) \\
%				&= \frac{1}{2}H_{{M}} (Y \mid X) + \frac{1}{2}H_{\bar{M}} (X \mid Y) +1 \\
%				&\le H^*_{\mathrm{con}} + 1
%\end{align*}
%By the symmetry of the matrix, we have 
%that~$H_{\bar{M}} (Y \mid X) =H_{\bar{M}} (X \mid Y)$.
%\end{proof}

\subsection{Other Deferred Lemmas and Proofs}

\begin{lemma}\label{lem:treeentropy}
Consider a node~$u$ connected directly to 
the root of a statically optimal self-adjusting 
\emph{ego-}$\BST(u)$, serving only requests to and from~$u$. 
 If~$\bar{p}$ is the empirical frequency distribution 
 of destinations and~$\bar{p}$ is the empirical frequency 
 distribution of sources, then 
 the amortized cost of routing in and adjusting~\emph{ego-}$\BST(u)$ is~$O(H(\bar{p}))$.
%If all requests in a demand~$\sigma$ 
%are from~$u$ or to~$u$ then the amortize cost of routing and adjusting is~$O(\max(H(\hat{X}_{\sigma}), H(\hat{Y}_{\sigma})))$.
\end{lemma}
\begin{proof}
A self-adjusting~\emph{ego-}$\BST(u)$ is originally designed 
to serve requests from the root to internal nodes. 
If the empirical frequency distribution on destinations 
(searched items) is~$\bar{p'}$, then the amortized cost
of~\emph{ego-}$\BST(u)$ is~$O(H(\bar{p'})$, which is optimal~\cite{splaytrees}. 
In our case, we also have routes from internal nodes in 
the tree toward the root. But for the self-adjusting~\emph{ego-}$\BST(u)$,
it does not matter if the request is~$(u,v)$ or~$(v,u)$:
the adjustments are the same, hence we can assume that each route 
request~$(v,u)$ is actually a~$(u,v)$ request, i.e.,  all requests
are from
the root of the tree. 
The new empirical frequency distribution on destinations
(when all requests are from root to destinations)
is also~$\bar{p}$. Therefore the results holds.
\end{proof}

\begin{lemma}[Helping Nodes]\label{lem:helping}
As long as the coordinator did not call~$\reset()$,
the size of the forwarding table of small nodes is at most $\Delta=6\theta$ and helping nodes are available if needed.
\end{lemma}
\begin{proof}
Only small nodes can be helper nodes.
A small node has a maximum degree of~$\theta$,
so it may need 12$c=$3$\theta$ ports in its forwarding table
for the working set.
For how many edges can a helper node be used
as a relay? 
Since the number of helper nodes is at least~$n/2$
(otherwise more than half  of the nodes have degree larger
than twice the average degree, which leads to a
contradiction) and since there are at most~$cn$ large-large edges, 
each helper node needs to help at most~$2c$ such edges.
Each helper node requires 6 ports (3 for each tree),
so in total it needs at most 3$\theta$ ports. Since the size of the 
forwarding table is 6$\theta$, there will always be a helper
node while the number of edges is less than~$cn$ which mean total size of all working sets is less than ~$2cn =\theta n/2$
\end{proof}

\end{document}